\RemoveFromHook{package/amsthm/after}[firstaid/aliascounter]
\documentclass[11pt]{article}
\usepackage{marginnote}

\usepackage[letterpaper,margin=1.00in]{geometry}
\usepackage{amsmath, amssymb, amsthm, amsfonts}
\usepackage{thmtools,thm-restate}
\usepackage{bbm}

\usepackage{ifthen}

\usepackage{amssymb}
\usepackage{amsthm}
\usepackage{newlfont} 
\usepackage{graphicx, amssymb, latexsym, amsfonts, amsmath, lscape, amscd,
	amsthm, color, epsfig} 
\usepackage{amsmath}
\usepackage{amsthm}
\usepackage{amscd}
\usepackage{graphicx}
\usepackage[bottom]{footmisc}
\usepackage{tikz}
\usepackage{algorithm}

\usepackage{cite}
\usepackage{appendix}
\usepackage{graphicx}
\usepackage{xcolor}
\usepackage{algorithm}
\usepackage{algorithm}
\usepackage[noend]{algpseudocode}
\usepackage{epstopdf}
\usepackage{wrapfig}
\usepackage{paralist}
\usepackage{wasysym}

\usepackage{dsfont}

\usepackage[textsize=tiny]{todonotes}

\usepackage{framed}
\usepackage[framemethod=tikz]{mdframed}
\usepackage[bottom]{footmisc}
\usepackage{enumitem}
\setitemize{noitemsep,topsep=3pt,parsep=3pt,partopsep=3pt}
\usepackage[font=small]{caption}
\usepackage{xspace}

\newtheorem{theorem}{Theorem}[section]
\newtheorem{lemma}[theorem]{Lemma}
\newtheorem{meta-theorem}[theorem]{Meta-Theorem}

\newtheorem{corollary}[theorem]{Corollary}

\newtheorem{definition}[theorem]{Definition}

\definecolor{darkgreen}{rgb}{0,0.5,0}
\definecolor{darkblue}{rgb}{0,0,0.6}
\usepackage{hyperref}
\hypersetup{
	unicode=false,          
	colorlinks=true,        
	linkcolor=darkblue,          
	citecolor=darkgreen,        
	filecolor=magenta,      
	urlcolor=cyan           
}

\usepackage[capitalize, nameinlink]{cleveref}
\Crefname{remark}{Remark}{Remarks}
\Crefname{observation}{Observation}{Observations}
\Crefname{Inequality}{Inequality}{Inequality}
\Crefname{Claim}{Claim}{Claim}
\Crefname{fact}{Fact}{Fact}

\algnewcommand\algorithmicswitch{\textbf{switch}}
\algnewcommand\algorithmiccase{\textbf{case}}

\algdef{SE}[SWITCH]{Switch}{EndSwitch}[1]{\algorithmicswitch\ #1\ \algorithmicdo}{\algorithmicend\ \algorithmicswitch}%
\algdef{SE}[CASE]{Case}{EndCase}[1]{\algorithmiccase\ #1}{\algorithmicend\ \algorithmiccase}%
\algtext*{EndSwitch}%
\algtext*{EndCase}%

\newcommand{\eps}{\varepsilon}

\newcommand{\CONGEST}{$\mathsf{CONGEST}$\xspace}
\newcommand{\LOCAL}{$\mathsf{LOCAL}$\xspace}

\newcommand{\poly}{\operatorname{\text{{\rm poly}}}}

\newcommand{\set}[1]{\left\{#1\right\}}

\newcommand{\OPT}{\mathsf{OPT}}

\newcommand{\utility}{\mathbf{u}}
\newcommand{\cost}{\mathbf{c}}

\newcommand{\FullOrShort}{full}

\ifthenelse{\equal{\FullOrShort}{full}}{
	
	\newcommand{\fullOnly}[1]{#1}
	\newcommand{\shortOnly}[1]{}

}{

	\newcommand{\fullOnly}[1]{}
	\newcommand{\IncludePictures}[1]{}
	
}

\providecommand{\calA}{\mathcal{A}}
\newcommand{\davg}{\overline{d}}

\title{Faster CONGEST Approximation Algorithms for Maximum Weight Independent Set in Sparse Graphs}
\author{Salwa Faour \\ \small{University of Freiburg} \\ \small{\texttt{salwa.faour@cs.uni-freiburg.de}}
\and Fabian Kuhn \\ \small{University of Freiburg} \\ \small{\texttt{kuhn@cs.uni-freiburg.de}}}
\date{}

\begin{document}

\maketitle
\begin{abstract}
The maximum independent set problem is a classic optimization problem in graph theory that has been extensively studied in the distributed setting. Although the problem is hard to approximate within reasonable factors in general, efficient approximation algorithms are known for several sparse graph families. In the present paper, we consider deterministic distributed \CONGEST algorithms for the weighted version of the problem in trees and graphs of bounded arboricity (i.e., hereditarily sparse graphs). All weights are nonnegative integers bounded by $n^{O(1)}$.

    For trees and $\log(n)/n<\varepsilon\leq1/8$, we prove that the task of deterministically computing a $(1-\varepsilon)$-approximate solution to the maximum weight independent set (MWIS) problem has a tight $\Theta(\log^*(n) / \varepsilon)$ complexity. The lower bound already holds on unweighted oriented paths. The upper bound holds even for unrooted trees.

    For graphs $G=(V,E)$ of arboricity $\beta>1$, we give two algorithms. If the sum of all node weights is $w(V)$, we show that for any $0<\varepsilon\leq1/2$, an independent set of weight at least $(1-\varepsilon)\cdot \frac{w(V)}{4\beta}$ can be computed in $O(\log^2(\beta/\varepsilon)/\varepsilon + \log(\beta/\eps) \cdot \log^* n)$ rounds. This result is obtained by a direct application of the local rounding framework of Faour, Ghaffari, Grunau, Kuhn, and Rozho\v{n} [SODA ‘23]. We further show that for any $0<\varepsilon\leq1/2$, an independent set of weight at least $(1-\varepsilon)\cdot\frac{w(V)}{2\beta+1}$ can be computed in $O(\log^3(\beta)\cdot\log(1/\varepsilon)/\varepsilon^2 \cdot\log n)$ rounds. For $\varepsilon=\omega(1/\sqrt{\beta})$, this significantly improves on a recent result of Gil [OPODIS ‘23], whose algorithm yields a $1/\lfloor(2+\varepsilon)\beta\rfloor$-approximation in $O(\beta/\eps\cdot\log n)$ rounds when its accuracy dependence is tracked. As an intermediate step to our result, we design an algorithm to compute an independent set of total weight at least $(1-\varepsilon)\cdot\sum_{v\in V}\frac{w(v)}{\deg(v)+1}$ in time $O(\log^3(\Delta)\cdot\log(1/\varepsilon)/\varepsilon + \log^* n)$, where $\Delta$ is the maximum degree of the graph and logarithmic runtime factors are lower-bounded by one.
\end{abstract}

\newif\ifshowreviewedits
\showrevieweditsfalse
\newcommand{\reviewchange}[1]{\ifshowreviewedits{\color{red}\hypersetup{linkcolor=red}#1}\else#1\fi}
\newcommand{\reviewnote}[2][0pt]{%
  \ifshowreviewedits
    \marginnote{\begingroup\setlength{\fboxsep}{2pt}%
      \colorbox{orange!22}{\parbox{48pt}{\raggedright\normalfont\sffamily\scriptsize\color{black}#2}}%
    \endgroup}[#1]%
  \fi
}
\section{Introduction and Related Work}\label{sec:intro}

Given a finite, simple, undirected graph $G=(V,E)$ with node weights $w:V\to \mathbb{Z}_{\geq0}$, the \emph{maximum weight independent set (MWIS)} problem asks for an independent set $I\subseteq V$ (i.e., a set $I$ of pairwise non-adjacent nodes) such that the total weight $w(I):=\sum_{v\in I} w(v)$ is maximized. The MWIS problem is a classic optimization problem on graphs and has also been studied extensively in the distributed setting. In the centralized setting, strong approximation hardness results are known for MWIS on general graphs~\cite{Hastad96}. We focus on distributed approximation in sparse graph families, where much better approximation guarantees are known. Before discussing our contributions, we summarize the existing literature on distributed approximation algorithms for MWIS.

In the standard setting for distributed graph algorithms, the graph $G=(V,E)$ on which the problem is defined is also the communication graph. There are $n=|V|$ nodes and we assume that each node is equipped with an $O(\log n)$-bit unique identifier. The nodes $V$ communicate over the edges in synchronous rounds. In the \LOCAL model~\cite{peleg00}, in each round, the nodes can exchange arbitrarily large messages over all the edges. In the more restricted \CONGEST model, in each round, each node can send a (possibly different) message of $O(\log n)$ bits to each neighbor. Throughout the paper, all input node and edge weights are nonnegative integers bounded by $n^c$ for a fixed constant $c$. Thus, each input weight has an $O(\log n)$-bit encoding. In subgraph and recursive computations, these encoding bounds and the message size refer to the number $n$ of nodes in the original communication network. The internal computation at the nodes is not restricted. Initially, the nodes do not know anything about the topology of the graph $G$. \reviewnote{Clarify model assumptions.}\reviewchange{We assume that all nodes know a common polynomial upper bound on $n$, the maximum degree $\Delta$, and the supplied accuracy and arboricity parameters.\footnote{\reviewchange{For example, in the algorithms of \Cref{sec:arboricity}, the size bound is used to bound the recursion depth in \Cref{algo:2beta+1approxMWIS}, and $\Delta$ determines the number of degree classes in the proof of \Cref{thm:carowei}.}}} At the end of an independent set algorithm, each node must know whether it is in the independent set.

\subsection{Distributed MWIS in General Graphs}

The most widely studied distributed independent set problem is the problem of computing a maximal independent set (MIS), i.e., an independent set that cannot be extended. The distributed complexity of computing an MIS has been studied intensively since the 1980s, e.g., \cite{alon86,Balliu0KO22-hideandseek,BEPS2016,BEK15,ghaffari16_MIS,Localrounding23,luby86,GhaffariGrunauFOCS24,kuhn16_jacm}. Randomized algorithms based on shattering have complexity $O(\log\Delta) + \tilde{O}(\log^{5/3} \log n)$ in the \LOCAL and $O(\log\Delta) + \tilde{O}(\log^3\log n)$ in the \CONGEST model~\cite{ghaffari16_MIS,Localrounding23,GhaffariGrunauFOCS24}.\footnote{We use the notation $\tilde{O}(\cdot)$ to hide polylogarithmic factors in the argument, i.e., $\tilde{O}(x)=x\cdot\poly\log x$.} Deterministic algorithms compute an MIS in $\tilde{O}(\log^{5/3} n)$ rounds in the \LOCAL and $\tilde{O}(\log^2\Delta\cdot\log n)$ rounds in the \CONGEST model~\cite{Localrounding23,GhaffariGrunauFOCS24}. For the \emph{maximum cardinality independent set (MCIS)} problem, i.e., for the unweighted version of the MWIS problem, an MIS directly gives a $1/\Delta$-approximation (where $\Delta$ denotes the maximum degree of the graph). In \cite{Bar-YehudaCGS17}, it is shown that in the \CONGEST model, a $1/\Delta$-approximation for the MWIS problem can be computed in time $O(\log W\cdot T_{\mathrm{MIS}})$, where $W$ denotes the ratio between the largest and smallest positive node weight (with $W=1$ if all weights are zero) and $T_{\mathrm{MIS}}$ is the time to compute an MIS. Subsequently, Kawarabayashi, Khoury, Schild, and Schwartzman~\cite{KawarabayashiKS20} give an algorithm to compute an independent set of weight at least $(1-\eps)\cdot \frac{w(V)}{\Delta+1}$ in time $O(1/\eps)$ times the time to compute an independent set of weight at least $\frac{w(V)}{c(\Delta+1)}$ for some constant $c>0$. For the latter problem, they give two algorithms. One of them reduces the problem to computing a single MIS of a (locally computable) subgraph of the input graph. The other one is a randomized algorithm that requires $\poly(\log\log n)$ rounds. The reduction from computing an independent set of weight $(1-\eps)\cdot \frac{w(V)}{\Delta+1}$ to the problem of computing an independent set of weight $\frac{w(V)}{c(\Delta+1)}$ for some constant $c>0$ is done by using the local-ratio technique of \cite{localratio}. Faour et al.~\cite{Localrounding23} adapt this algorithm in two ways. First, they use $O(\log(1/\eps))$, rather than $O(1/\eps)$, instances computing an independent set of weight at least $\frac{w(V)}{c(\Delta+1)}$. Second, their local rounding technique gives a deterministic $O(\log^2\Delta + \log^* n)$-round \CONGEST algorithm computing an independent set of weight at least $\frac{w(V)}{4(\Delta+1)}$.

For general graphs, better approximations can in particular be achieved in the \LOCAL model. Randomized MPX-style low-diameter clustering can be used to obtain pairwise non-adjacent clusters of diameter $O(\log(n)/\eps)$ such that every node is left outside the clusters with probability at most $\eps$~\cite{MPX13}. Computing an optimal independent set independently in every cluster therefore yields an independent set $I$ with $\mathbb{E}[w(I)]\geq(1-\eps)\OPT$. Deterministically, a $(1-\eps)$-approximation can instead be obtained from the sequential ball-growing approach of Ghaffari, Kuhn, and Maus~\cite[Section~7]{ghaffari2017complexity} by using a network decomposition to schedule its locally dependent steps in the \LOCAL model. With the currently best decomposition algorithm by Ghaffari and Grunau~\cite{GhaffariGrunauFOCS24}, this adds a $\tilde{O}(\log^2 n)$ overhead, i.e., roughly two additional $\log n$ factors.

\subsection{Small Neighborhood Independence}

The neighborhood independence $\theta$ of a graph is the maximum number of pairwise non-adjacent neighbors of any node of the graph. In a graph with neighborhood independence $\theta$, an MIS directly gives a $1/\theta$-approximation for the MCIS problem (see, e.g., \cite{Hochbaum83}).\footnote{To see this, one can assign each node $v$ of some optimal independent set to a neighboring node in a given MIS (or to itself if $v$ is in the given MIS). The number of nodes of the optimal independent set that are assigned to each node in the MIS is then at most $\theta$.} In weighted graphs of neighborhood independence $\theta$, there is a deterministic \CONGEST algorithm for computing a $(1-\eps)/\theta$-approximation for MWIS in $O(\log^2(\Delta W)\cdot\log(1/\eps)+\log^* n)$ rounds~\cite{Localrounding23}.
As a special case of this, we obtain a $(1-\eps)/2$-approximation to the weighted maximum matching problem on graphs, and more generally a $(1-\eps)/r$-approximation for weighted maximum matching on hypergraphs of rank $r$. In addition, some classes of geometric graphs that are used to model wireless networks, such as unit disk graphs, also have neighborhood independence $O(1)$ and we therefore also get constant approximations for MWIS for those graphs. We note that in those geometric graphs, an MIS can be computed in only $O(\log^* n)$ rounds~\cite{schneider08}.

The distributed maximum (weighted) matching problem has also been studied by itself quite intensively, e.g.,~\cite{ahmadi18,Bar-YehudaCGS17,FaourFK21,Fischer20,FischerMU22,lotker15,kuhn16_jacm}. In the \CONGEST model, a $(2/3-\eps)$-approximation in general graphs and a $(1-\eps)$-approximation in bipartite graphs can be computed deterministically in $O\big(\frac{\log (\Delta W)}{\eps^2} + \frac{\log^2\Delta + \log^* n}{\eps}\big)$ rounds, where $W$ denotes the ratio between the largest and smallest positive edge weight (again setting $W=1$ if all weights are zero)~\cite{ahmadi18}.  It is further known that in general graphs, a $(1-\eps)$-approximation to the maximum cardinality matching problem can be computed in $\poly(\log(n)/\eps)$ (randomized) \CONGEST rounds~\cite{FischerMU22}.

\subsection{Distributed MWIS in Sparse Graphs}\label{sec:intro_sparse}

We now get to the case of sparse graphs, which are the focus of this paper. It has been shown by Tur\'an~\cite{turan41} that every graph of average degree $\davg$ has an independent set of size at least $n/(\davg+1)$. A natural generalization of this is the so-called Caro-Wei bound~\cite{caro79,Wei1981generalizedturan}, which states that every graph $G$ has an independent set of size
\begin{equation}\label{eq:carowei}
    \alpha(G) \geq \mathsf{CaroWei}(G) := \sum_{v\in V}\frac{1}{\deg(v)+1} \geq \frac{n}{\davg+1}. 
\end{equation}
The second inequality follows from the Cauchy--Schwarz inequality. As observed by Boppana, the Caro-Wei bound can be obtained in expectation by a simple random process that can be implemented in a single round in the \CONGEST model~\cite{BoppanaHR18,HalldorssonKonrad18}. Every node takes an independent sample from the uniform distribution on $[0,1]$ and joins the independent set if and only if its sample is smaller than those of all its neighbors. A node $v$ is selected with probability exactly $1/(\deg(v)+1)$, which gives \eqref{eq:carowei} by linearity of expectation. With random numbers from a sufficiently large $O(\log n)$-bit domain and IDs used to break ties, the same procedure attains the bound up to a negligible additive loss. This algorithm is equivalent to running a single phase of Luby's classic MIS algorithm~\cite{luby86}. In \cite{BoppanaHR18}, it is shown that the Caro-Wei bound \eqref{eq:carowei} is a $\frac{4\cdot(\sqrt{2}-1)}{\davg+2}\approx \frac{1.657}{\davg+2}$-approximation of the MCIS problem. For sparse graphs, we often want to express the approximation factor as a function of the arboricity $\beta$ of the graph. The arboricity of a graph is defined as the minimum number of forests into which its edges can be partitioned. Since a graph of arboricity $\beta$ has average degree $\davg<2\beta$, the result of \cite{BoppanaHR18} implies that \eqref{eq:carowei} is a $\frac{2\cdot(\sqrt{2}-1)}{\beta+1}\approx \frac{0.828}{\beta+1}$-approximation of the MCIS problem. In \cite{Localrounding23}, it is shown that there is a deterministic \CONGEST algorithm with round complexity $O\big(\frac{\log^2(\Delta/\eps)}{\eps} + \log^* n\big)$ that comes within a factor $(1/2-\eps)$ of the bound of \eqref{eq:carowei} in the unweighted case. Its weighted result concerns the different quantity $\sum_{v:\,w(v)>0} w(v)^2/(w(v)+w(N(v)))$. In contrast, \Cref{thm:carowei} obtains a $(1-\eps)$-fraction of $\sum_v w(v)/(\deg(v)+1)$. The two weighted quantities need not dominate one another; our factor improvement therefore concerns the unweighted Caro--Wei benchmark. For fixed accuracy, it costs an additional logarithmic factor in $\Delta$.

The bound \eqref{eq:carowei} and also its weighted generalizations do not allow one to approximate the weighted independent set problem within a factor linear in the arboricity $\beta$. For a different weighted generalization, \cite{BoppanaHR18} obtains an approximation whose asymptotic leading constant is $3/\sqrt[3]{4}\approx 1.89$ in the reciprocal-degree guarantee; its finite-degree guarantee is given separately in that paper. The first paper that achieves an efficient $1/O(\beta)$-approximation for the MWIS problem in the \CONGEST model is \cite{KawarabayashiKS20}. They achieve a $(1-\eps)/(8\beta)$-approximation by first decomposing the graph into $O(\log n)$ layers of degree at most $4\beta-1$ and by running an algorithm to obtain an independent set of weight $(1-\eps)\cdot w_i(V_i)/(4\beta)$ within each layer $V_i$ (and where $w_i$ is an appropriate weight function that is used in layer $V_i$). The overall round complexity is $O(\log n)$ times the time for computing the independent set of weight $(1-\eps)\cdot w_i(V_i)/(4\beta)$ in each layer. By using the randomized algorithm of \cite{KawarabayashiKS20}, the overall complexity of the resulting randomized $(1-\eps)/(8\beta)$-approximation for the MWIS problem is $O\big(\frac{\log n\cdot \poly\log\log n}{\eps}\big)$ and by using the deterministic algorithm of \cite{Localrounding23}, the overall complexity of the resulting deterministic $(1-\eps)/(8\beta)$-approximation for the MWIS problem is $O\big(\log^2\beta\cdot \log(1/\eps)\cdot \log n\big)$. In fact, at the cost of an additional $1/\eps$-factor in the round complexity, the approximation quality of the algorithm of \cite{KawarabayashiKS20} can be improved to $(1-\eps)/(4\beta)$. Gil~\cite{YuvalGil24} improves the approximation quality by a factor of two using an LP-based layer-processing algorithm. For constant $\eps$, the published theorem gives a deterministic $O(\beta\log n)$-round \CONGEST algorithm computing a $1/\lfloor(2+\eps)\beta\rfloor$-approximation; tracking the layer count gives $O(\beta\log n/\eps)$ rounds for variable $\eps$.

Two important classes of sparse graphs are trees and minor-closed graph families (e.g., planar graphs). In \cite{lenzen08,CzygrinowHW08}, it was shown that on cycles (and paths), deterministically computing a constant approximation for the maximum cardinality independent set problem requires $\Omega(\log^* n)$ rounds. In every graph of arboricity $\beta=O(1)$, a constant MCIS approximation follows from a standard bounded-degree reduction. Half the nodes have degree at most $4\beta$ and an independent set of size $\Omega(n)$ among the nodes of degree at most $4\beta$ can be computed by computing an MIS in time $O(\beta+\log^* n)$~\cite{BEK15}. For planar graphs (and thus also for trees), it is shown in \cite{CzygrinowHW08} that in the \LOCAL model, even for the MWIS problem, a $(1-\eps)$-approximation can be computed in time $\poly(\eps^{-1}\cdot\log^* n)$. One of our goals is to determine the precise dependence on $\eps$ for trees and to obtain the same bound in the more restrictive \CONGEST model.

\subsection{Our Contributions}
\label{sec:contributions}

We give deterministic \CONGEST approximation algorithms for MWIS in sparse graph families, together with matching tree lower bounds, and outline the main ideas behind them. We express approximation guarantees as factors at most one for maximization problems and at least one for minimization problems. We state our upper bounds for $0<\eps\leq 1/2$; for larger error parameters, running with accuracy $1/2$ gives the requested approximation, with the runtime evaluated at $1/2$. In round-complexity bounds, all logarithmic factors are lower-bounded by one, so $\log x$ denotes $\log_2\max\{2,x\}$ there, and the same lower bound applies to $\log^*x$.

\subsubsection{Approximating MWIS in Trees}
\label{sec:contribTrees}
We start by establishing a tight bound for computing a $(1-\eps)$-approximation for MWIS in trees.

\begin{theorem}\label{thm:mainTree}
There exists a deterministic algorithm that, given a tree network and a parameter $0<\eps\leq 1/2$, computes a $(1-\eps)$-approximate solution to the maximum weight independent set problem in
$O\!\left(\tfrac{\log^{*} n}{\eps}\right)$ rounds in the \CONGEST model. Moreover, for
$\tfrac{\log n}{n}<\eps\leq\tfrac18$, every deterministic \LOCAL algorithm for the unweighted problem on oriented paths requires
$\Omega\!\left(\tfrac{\log^{*} n}{\eps}\right)$ rounds.
\end{theorem}

The upper bound is based on ideas that were developed by Czygrinow, Ha{\'{n}}{\'{c}}kowiak, and Wawrzyniak for planar graphs in~\cite{CzygrinowHW08}. The core idea of the algorithm of \cite{CzygrinowHW08} is the following. First, one defines a weight function on edges, where for $\set{u,v}\in E$, we set $w(\set{u,v}):=\min\set{w(u), w(v)}$. The algorithm of \cite{CzygrinowHW08} then computes a clustering of the nodes of $V$ such that the total weight of the edges connecting nodes in different clusters is at most an $\eps/2$-fraction of the total weight of all the edges. By computing an optimal weighted independent set within each cluster, one can then obtain a $(1-\eps)$-approximate MWIS solution of the given instance. The independent sets of the clusters are combined by taking the union and removing the lower-weight endpoint of every intercluster edge, breaking ties by ID. The clustering algorithm of \cite{CzygrinowHW08} obtains clusters of diameter $\poly(\eps^{-1})$ in time $\poly(\eps^{-1})\cdot\log^* n$ in the \LOCAL model. In trees, the clustering algorithm of \cite{CzygrinowHW08} can be implemented in the \CONGEST model. Further, in trees, an optimal solution to MWIS inside each cluster can be computed in time linear in the diameter of the cluster by using a straightforward dynamic programming algorithm.

The core challenge in proving the upper bound of \Cref{thm:mainTree} is to obtain a clustering with better guarantees. As our main technical contribution of this part, we show that in $O\big(\frac{\log^* n}{\eps}\big)$ rounds, one can compute a clustering where the total weight of the intercluster edges is at most an $\eps/2$-fraction of the total weight of all the edges and where the maximum cluster diameter is at most $O(1/\eps)$.

To obtain this, we adapt the algorithm of \cite{CzygrinowHW08} in the following way. The algorithm of \cite{CzygrinowHW08} consists of basic steps in which the cluster diameter grows by a constant factor and the total weight of the intercluster edges shrinks by a constant factor. However, the cluster diameter grows by a factor that is larger than the factor by which the total weight of the intercluster edges decreases. In our algorithm, we show that in trees, one can interleave those basic steps with steps in which clusters with a large diameter are split into clusters of smaller diameter, while increasing the overall weight of the intercluster edges by only a controlled amount.

The lower bound of \Cref{thm:mainTree} is obtained as follows. In \cite{CzygrinowHW08,lenzen08}, it is shown that deterministically obtaining any constant approximation for maximum cardinality independent set on unoriented rings (i.e., for the unweighted version of MWIS) requires $\Omega(\log^* n)$ rounds. It is not hard to obtain the same lower bound also for oriented paths. Our lower bound is then obtained by a reduction from the problem of computing a constant approximation on oriented paths. For the reduction, each edge of a path is subdivided into a path of length $2k+1$, where $k=c/\eps$ for a sufficiently small constant $c$. \reviewnote{Explain the error range, shorten for page limit.}\reviewchange{The lower cutoff $\eps>\log(n)/n$ keeps the original path large enough to yield the claimed round lower bound, and the upper cutoff $\eps\leq1/8$ ensures a nontrivial subdivision in the proof.} One can then show that any $(1-\eps)$-approximation of maximum cardinality independent set on this virtual path can be transformed into a constant approximation of the problem on the original path.

\paragraph{Approximating Unweighted Dominating Set, Vertex Cover, and Matching} \reviewnote{Shorten for page limit.}\reviewchange{Our clustering also gives $(1\pm\eps)$-approximations for these three problems in $O(\log^* n/\eps)$ rounds. In a tree with at least three nodes, each optimum is within a constant factor of the number of non-leaf nodes. Clustering the subtree they induce therefore bounds the loss from intercluster edges relative to each optimum. As in \Cref{thm:mainTree}, the round bound is tight for $\log(n)/n<\eps\leq1/8$.}

\subsubsection{Approximating MWIS as a Function of the Arboricity}
\label{sec:contribArboricity}

For the remainder of this section, we assume that we are given a graph $G=(V,E)$ of arboricity $\beta$. In \cite{KawarabayashiKS20}, it is shown that there is a randomized \CONGEST algorithm to compute an independent set of weight at least $(1-\eps)w(V)/(8\beta)$ in time $O\big(\frac{\log n\cdot\poly\log\log n}{\eps}\big)$. In combination with an algorithm from \cite{Localrounding23}, the same guarantee can be obtained deterministically in time $O\big(\log^2\beta\cdot\log(1/\eps)\cdot\log n\big)$. At the cost of an additional factor $1/\eps$ in the runtime, the approximation guarantee can be improved to $(1-\eps)/(4\beta)$. In the following, we show that for moderately small $\beta$ both the approximation quality and the runtime can be improved significantly.

\begin{restatable}{theorem}{restatefirst}\label{thm:arboricity_rounding}There exists a deterministic algorithm that, for any graph $G=(V,E)$ with arboricity $\beta \geq 1$ and node weight function $w:V\to \mathbb{Z}_{\geq0}$ and any $0<\eps\leq 1/2$, computes an independent set $I$ of weight $w(I) \geq (1-\eps)\cdot \frac{w(V)}{4 \beta}$ in $O\big(\frac{\log^2(\beta/\eps)}{\eps}+ \log(\beta/\eps) \cdot \log^* n\big)$ rounds in the \CONGEST model.
\end{restatable}

The theorem follows from a black-box application of the local rounding framework of \cite{Localrounding23}. More concretely, we show that \Cref{thm:arboricity_rounding} follows from the utility--cost rounding lemma used there. \Cref{thm:arboricity_rounding} implies a $(1-\eps)/(4\beta)$-approximation of MWIS. Gil's published theorem~\cite{YuvalGil24} gives, for constant $\eps$, a deterministic $O(\beta\log n)$-round \CONGEST algorithm computing a $1/\lfloor(2+\eps)\beta\rfloor$-approximation; tracking its peeling layers gives $O(\beta\log n/\eps)$ rounds for variable $\eps$. After a constant-factor adjustment of the accuracy parameters, our next theorem gives at least as strong an approximation in the regime considered below. The ratio of its runtime bound to Gil's is $O(\log^3\beta\cdot\log(1/\eps)/(\beta\eps))$, which tends to zero when $\eps=\omega(1/\sqrt{\beta})$.

\begin{theorem}\label{thm:arboricity}
 There exists a deterministic algorithm that, for any graph $G=(V,E)$ with arboricity $\beta\geq 1$ and node weight function $w:V\to \mathbb{Z}_{\geq0}$ and any $0<\eps\leq 1/2$, computes an independent set $I$ of weight $w(I) \geq (1-\eps)\cdot \frac{w(V)}{2\beta+1}$ in $O\big(\frac{\log^3(\beta)\cdot\log(1/\eps)\cdot\log n}{\eps^2}\big)$ rounds in the \CONGEST model.
\end{theorem}

A major step towards proving \Cref{thm:arboricity} is an efficient algorithm to compute an independent set of weight arbitrarily close to a natural weighted generalization of the Caro-Wei bound \eqref{eq:carowei}. Specifically, we prove the following theorem.

\begin{theorem}\label{thm:carowei}
There exists a deterministic algorithm that for any graph $G=(V,E)$ with node weights $w:V\to \mathbb{Z}_{\geq0}$ and maximum degree $\Delta$, equipped with an initial proper $\xi$-coloring, where $\xi\leq n^{O(1)}$, and for any $0<\eps\leq 1/2$, computes an independent set $I$ of total weight
$w(I) \geq (1-\eps)\cdot\sum_{v\in V}\frac{w(v)}{\deg(v)+1}$ in $O\big(\frac{\log(1/\eps)}{\eps}\cdot \log^3\Delta + \log^* \xi \big)$ rounds in the \CONGEST model.
\end{theorem}

As a direct application of \Cref{thm:carowei}, consider the maximum cardinality independent set problem in graphs of arboricity $\beta$. As proven in \cite{BoppanaHR18} and discussed in \Cref{sec:intro_sparse}, the Caro-Wei bound \eqref{eq:carowei} provides a $\frac{2\cdot(\sqrt{2}-1)}{\beta+1}$-approximation to the maximum cardinality independent set problem in graphs of arboricity $\beta$. \Cref{thm:carowei} therefore implies that a $(1-\eps)\cdot\frac{2\cdot(\sqrt{2}-1)}{\beta+1}$-approximation can be computed in $O\big(\frac{\log^3(\Delta)\cdot\log(1/\eps)}{\eps}+\log^* n\big)$ rounds in the \CONGEST model. This can be further improved slightly. Because in a graph $G$ of arboricity $\beta$, for every $x\geq 2$, the number of nodes of degree more than $x\cdot \beta$ is at most $2n/x$ and the independence number is at least $\alpha(G)\geq n/(2\beta)$, the number of nodes of degree more than $c\cdot \beta^2/\eps$ is at most $4\eps\alpha(G)/c$. By applying \Cref{thm:carowei} with accuracy $\eps/2$ to the subgraph induced by the nodes of degree at most $8\beta^2/\eps$, we thus obtain the following corollary.

\begin{corollary}\label{cor:carowei}
 There exists a deterministic algorithm that, for any graph $G=(V,E)$ of arboricity $\beta \geq 1$ and any $0<\eps\leq 1/2$, computes a $ (1-\eps)\cdot \frac{2\cdot(\sqrt{2}-1)}{\beta+1}$-approximation to the maximum cardinality independent set problem in $O\big(\frac{\log^3(\beta/\eps)\cdot\log(1/\eps)}{\eps}+\log^* n\big)$ rounds in the \CONGEST model.
\end{corollary}

\begin{proof}
    Let $r_\beta:=2(\sqrt{2}-1)/(\beta+1)$. Since arboricity is hereditary, every nonempty subgraph of $G$ has average degree less than $2\beta$ and hence a node of degree at most $2\beta-1$. Thus, $G$ is $(2\beta-1)$-degenerate and therefore $2\beta$-colorable, which implies $\alpha(G)\geq n/(2\beta)$. Let $B:=\{v\in V:\deg_G(v)>8\beta^2/\eps\}$ and $H:=G[V\setminus B]$. Since $\sum_{v\in V}\deg_G(v)\leq2\beta n$, we have
    \[
        |B|\leq\frac{\eps n}{4\beta}\leq\frac\eps2\alpha(G),
        \qquad
        \alpha(H)\geq\alpha(G)-|B|\geq(1-\eps/2)\alpha(G).
    \]
    Apply \Cref{thm:carowei} to $H$ with accuracy $\eps/2$ and unit weights, including isolated nodes locally. The arboricity of $H$ is at most $\beta$, so the returned independent set satisfies
    \[
        |I|\geq(1-\eps/2)\mathsf{CaroWei}(H)
        \geq(1-\eps/2)r_\beta\alpha(H)
        \geq(1-\eps/2)^2r_\beta\alpha(G)
        \geq(1-\eps)r_\beta\alpha(G).
    \]
    The degree cutoff requires $O(1)$ rounds, and $H$ has maximum degree at most $8\beta^2/\eps$. Substituting this bound and accuracy $\eps/2$ into the runtime of \Cref{thm:carowei} gives the claimed round complexity.
\end{proof}

The proof appears below. We next sketch how \Cref{thm:carowei} is proven and how it is used to prove \Cref{thm:arboricity}.
The starting point to proving \Cref{thm:carowei} is an $O\big(\log^2\Delta\cdot\log(1/\eps)+\log^* n\big)$-round algorithm from \cite{Localrounding23} to compute an independent set of weight at least $(1-\eps)\cdot w(V)/(\Delta+1)$ in a graph of maximum degree $\Delta$. The idea now is to divide the nodes of $G$ into $O(\log(\Delta+1)/\eps)$ degree classes so that the degrees within one class differ by at most a factor $1+O(\eps)$. We then want to essentially apply the algorithm from \cite{Localrounding23} separately for each degree class, starting from the small degrees. This would work almost directly for unweighted graphs. To make it work, one just has to include all the still available nodes from the lower degree classes when computing the independent set of some degree class.

\reviewnote{Shorten local-ratio overview.}\reviewchange{To combine independent sets across degree classes, we use the local-ratio method~\cite{localratio,Bar-YehudaCGS17,KawarabayashiKS20,Localrounding23}. After selecting an independent set in one class, we reduce each remaining node's weight by the total weight of its selected neighbors, down to a minimum of zero. These reductions account for earlier selections that must be discarded when a neighbor is selected later. Iterating this step over the $O(\log(\Delta+1)/\eps)$ classes yields \Cref{thm:carowei}.}

We next discuss our algorithm to obtain the bound claimed by \Cref{thm:arboricity}. Note that since the average degree of a graph of arboricity $\beta$ can be almost $2\beta$, we might have to compute an independent set of weight close to $w(V)$ divided by the average degree. There are graphs of arboricity $\beta$ for which this is best possible (e.g., unweighted cliques) and we therefore in some cases cannot afford to allow too much slack when iteratively computing an independent set. A standard tool to algorithmically deal with graphs of bounded arboricity is the so-called $H$-decomposition as introduced in \cite{barenboim08}. Because the average degree of a graph with arboricity $\beta$ is at most $2\beta$, the number of nodes of degree more than $(2+\eps)\beta$ is at most $(1-\Theta(\eps))\cdot n$. When peeling off all nodes of degree at most $(2+\eps)\beta$, we therefore remove at least an $\Omega(\eps)$ fraction of all nodes. Repeating this idea gives a decomposition of the graph into $O(\log(n)/\eps)$ layers so that the nodes in each layer have at most $(2+\eps)\beta$ neighbors in the same layer and higher layers. We are not aware of an algorithm to compute an independent set of weight close to $w(V)/(2\beta)$ that does not at least implicitly use an $H$-decomposition (which in most cases leads to a time complexity linear in the number of layers, whose upper bound is $O(\log(n)/\eps)$).

The high-level idea of our approach is the following. We again use the framework of \cite{localratio,Bar-YehudaCGS17,KawarabayashiKS20} to iteratively construct an independent set in the weighted setting. Because the nodes in the lowest layer are always the ones of bounded degree, we start our iterative construction at the lowest layer and work our way up the decomposition. Some of the main challenges with this approach already appear for the unweighted case and in the following high-level discussion, we therefore focus on this case. Let $V_0$ be the nodes of the lowest layer of our decomposition. All nodes in $V_0$ have degree at most $(2+\eps)\beta$ and we can therefore compute an independent set $I_0\subseteq V_0$ of size close to $|V_0|/(2\beta)$ (or even close to the Caro-Wei bound \eqref{eq:carowei} within $V_0$). When moving to the next layer, we cannot, however, include any remaining nodes in $V_0$ because this might create high-degree nodes in the layer under consideration. We do not need $I_0$ to dominate most of $V_0$. Instead, we need the incidence budget $\sum_{v\in I_0}(\deg_0(v)+1)$, where $\deg_0(v)$ is the degree of $v$ in $G[V_0]$, to be close to $|V_0|$. This budget supports the charging argument that combines the current layer with the residual instance. To obtain it, define $\overline{w}(v):=\deg_0(v)+1$ for every $v\in V_0$. Then \Cref{thm:carowei} computes an independent set $I_0$ with $\overline{w}(I_0)\approx\sum_{v\in V_0}\overline{w}(v)/(\deg_0(v)+1)=|V_0|$. In \Cref{sec:arboricity}, we provide the details and show that the same idea also directly works in the weighted case by using $\overline{w}(v):=(\deg_0(v)+1)w(v)$.

\subsection{Mathematical Preliminaries and Notation}\label{sec:prelim}

We next define the notation used throughout the paper. Consider a graph $G=(V,E)$. An \emph{independent set} $I\subseteq V$ is a set of pairwise non-adjacent nodes. The size of a largest independent set of $G$ is denoted by the \emph{independence number} $\alpha(G)$. The \emph{arboricity} $\beta$ of a graph $G$ is the minimum number of forests whose edge sets partition $E$. Consequently, every nonempty subgraph $H$ of a graph of arboricity at most $\beta\geq1$ has average degree less than $2\beta$.

For a node $v\in V$, we use $N(v)$ to denote the set of neighbors of $v$ and we use $\deg(v)=|N(v)|$ to denote the degree of $v$. We further define $N^+(v):=\set{v}\cup N(v)$ as the closed neighborhood of $v$. For a set of nodes $S\subseteq V$, we use $N(S)$ to denote the set of neighbors of $S$, i.e., $N(S)=(\bigcup_{v\in S} N(v))\setminus S$. We further use $N^+(S):=S\cup N(S)$ to denote the closed neighborhood of $S$. If we are given node weights $w:V\to\mathbb{Z}_{\geq0}$, we use the following notation. For a set $S\subseteq V$ of nodes, we define $w(S):=\sum_{v\in S} w(v)$. If we have edge weights $w:E\to\mathbb{Z}_{\geq0}$, we similarly define $w(F):=\sum_{e\in F} w(e)$ for a set $F\subseteq E$ of edges. Logarithms are to base $2$ unless specified. The convention of lower-bounding logarithmic factors by one applies only to round-complexity bounds, not to thresholds or algorithmic definitions. Edgeless graphs can be solved locally by including every node in the independent set; this also handles arboricity $\beta=0$.

The integer-weight convention also bounds the messages used by our algorithms. Write $B\leq n^c$ for the largest input weight. Tree dynamic-programming values and cluster weight sums are at most $nB$. Local-ratio residuals remain nonnegative integers and never exceed the corresponding input weights. In \Cref{algo:2beta+1approxMWIS}, the rescaled weights $(\deg_0(v)+1)w(v)$ are at most $nB$ and their sums are at most $n^2B$. Only unscaled residual weights are passed to recursive calls, so this rescaling does not accumulate with the recursion depth. All these quantities therefore have $O(\log n)$-bit encodings.

\subsection{Organization of the Paper}
The remainder of the paper is organized as follows. In \Cref{sec:trees}, we prove \Cref{thm:mainTree} and in \Cref{sec:arboricity}, we show \Cref{thm:arboricity_rounding}, then prove \Cref{thm:carowei} and \Cref{thm:arboricity}.

\section{Tight Approximation Bounds in Trees}
\label{sec:trees}

In the following, we prove \Cref{thm:mainTree} and thus the fact that the deterministic complexity of computing a $(1-\eps)$-approximation for MWIS in trees (and also a $(1\pm \eps)$-approximation for unweighted Minimum Dominating Set, Minimum Vertex Cover, and Maximum Matching in trees) in the \CONGEST model is $\Theta\big(\frac{\log^* n}{\eps}\big)$ for $\log(n)/n<\eps\leq1/8$. We start by proving the upper bound.

\subsection{Computing the Clustering}

As described in \Cref{sec:contribTrees}, the core part of the algorithm is to compute a good clustering of an edge-weighted tree where the total weight of the intercluster edges is at most $\lambda\cdot w(E)$ and where the cluster diameter is $O(1/\lambda)$ for an appropriate parameter $\lambda>0$. To formally argue about the quality of a clustering, we make the following definition.

\begin{definition}[$(d,\lambda)$-Clustering]\label{def:clustering}
    Let $d\geq 0$ be an integer and let $\lambda\in [0,1]$. A $(d,\lambda)$-clustering of an edge-weighted tree $T=(V,E,w)$ is a partition of $V$ into connected clusters $V_1,\dots,V_s$ such that for all $V_i$, the induced subtree $T[V_i]$ has diameter at most $d$ and the total weight of intercluster edges is at most $\lambda\cdot w(E)$.
\end{definition}

We represent a clustering distributively as follows. Every node knows its $O(\log n)$-bit cluster identifier and which incident edges are internal to its cluster, and every cluster has a distinguished root and parent pointers in its induced subtree. This representation, including the cluster diameter, can be established or re-established after a merge or split in $O(d+1)$ rounds for clusters of diameter at most $d$.

In \cite{CzygrinowHW08}, it is shown that given some clustering, one can merge some clusters such that in the resulting clustering, the maximum cluster diameter only increases by a constant factor and such that the total weight of the intercluster edges decreases by a constant factor. More specifically, we prove the following lemma, which was established in \cite{CzygrinowHW08} for the \LOCAL model. The proof largely follows the arguments in \cite{CzygrinowHW08}; however, we include it in the appendix for completeness.

\begin{restatable}{lemma}{restateclustering}\label{lemma:clusterincrease}~\cite{CzygrinowHW08}
There exists a deterministic $O((d+1)\cdot \log^* n)$-round \CONGEST algorithm that, given an edge-weighted tree $T=(V,E,w)$ with a given $(d,\lambda)$-clustering, computes an $(9d+8, 3\lambda/4)$-clustering of $T$.
\end{restatable}

\begin{proof}
  If the initial clustering has only one cluster, we return it unchanged. Otherwise, we first describe the algorithm abstractly and then give its deterministic \CONGEST implementation. Contract every cluster of the initial $(d,\lambda)$-clustering to one node, and let $T'=(V',E')$ be the resulting tree. Its edges are precisely the intercluster edges.\footnote{Two contracted cluster nodes are joined by at most one edge because $T$ is acyclic.} Hence $w(E')\leq\lambda w(E)$. Each node $v\in V'$ with at least one incident edge marks a heaviest such edge $e_v$, breaking ties by a fixed rule based on endpoint and cluster IDs; an isolated node marks nothing and remains a singleton cluster.

  Let $E'_M$ be the set of marked edges. Consider an arbitrary orientation of $T'$ with outdegree at most one, and let $w_v$ be the weight of $v$'s outgoing edge, or zero if none exists. Then $\sum_{v\in V'}w_v=w(E')$ and $w(e_v)\geq w_v$ whenever $e_v$ is defined; interpret $w(e_v)$ as zero for an isolated node. Since every edge is marked by at most two nodes, $w(E'_M)\geq\frac12\sum_{v\in V'}w(e_v)\geq w(E')/2$. Orient an edge of $E'_M$ away from its unique marking endpoint; if both endpoints marked it, orient it by ID. This roots every component of the marked forest $T'[E'_M]$.

  We next break $T'[E'_M]$ into trees of constant diameter. Compute a $3$-coloring of this rooted forest. For $v\in V'$, let $e_p(v)$ be its parent edge, if present, and let $C(v)$ be its child edges. A color-$1$ node removes $C(v)$ if it has a parent and $w(e_p(v))>w(C(v))$; otherwise, it removes $e_p(v)$ if present. For a color-$2$ node, let $C_3(v)$ be the edges to color-$3$ children. If its parent has color $3$, it removes $C_3(v)$ when $w(e_p(v))>w(C_3(v))$ and otherwise removes $e_p(v)$. It takes no action when its parent does not have color $3$. Every edge can be removed only by its color-$1$ endpoint or, for a color-$2$--color-$3$ edge, by its color-$2$ endpoint. Each node removes at most half the weight for which it is responsible, so at most half of $w(E'_M)$ is removed. Let $E''_M$ be the remaining set. Then $w(E''_M)\geq w(E')/4$. The new clusters are the connected components of the spanning graph $\bigl(V,(E\setminus E')\cup E''_M\bigr)$, including isolated nodes.

  The new intercluster-edge set is $E'\setminus E''_M$, whose weight is at most $3w(E')/4\leq3\lambda w(E)/4$. It remains to bound the new cluster diameter. Every directed path in $T'[E''_M]$ has at most four edges: a color-$1$ node removes either its outgoing edge or all incoming edges and hence cannot be internal, while a color-$2$ node cannot be at distance at least two from both endpoints because it would then have incoming and outgoing color-$3$ neighbors. Thus all nodes at distance at least two from both endpoints have color $3$, and there can be at most one. An undirected path therefore has at most eight edges. It traverses at most nine old clusters, each of diameter at most $d$, and at most eight edges of $E'$, giving diameter at most $9d+8$.

  For the implementation, each cluster maintains an identifier, a rooted communication tree, and its boundary information; a broadcast or aggregate within a cluster takes $O(d+1)$ rounds. Thus each cluster finds and marks a heaviest incident quotient edge in $O(d+1)$ rounds. We compute the $3$-coloring with the Cole--Vishkin algorithm~\cite{cole86}. In each of its $O(\log^*n)$ rounds, a quotient node broadcasts one common intermediate color to all its children. This specific communication pattern, rather than a simulation of an arbitrary \LOCAL round, can be routed through each cluster in $O(d+1)$ rounds. The final deletion step also takes $O(d+1)$ rounds. The total deterministic \CONGEST complexity is therefore $O((d+1)\log^*n)$.
\end{proof}

The algorithm of \cite{CzygrinowHW08} uses the construction of \Cref{lemma:clusterincrease} (with an additional step that gives slightly better constants). Because the factor by which the maximum cluster diameter increases is larger than the factor by which the weight of the intercluster edges decreases, this yields a $(\poly(1/\eps), \eps/2)$-clustering, which can then be used in \Cref{lemma:clusteringIS}. In \cite{CzygrinowHW08}, the algorithm is described for the family of planar graphs. In the special case of trees, we can do better. The following lemma shows that if the maximum cluster diameter is too large, it can be reduced at the cost of increasing the weight of the intercluster edges by a controlled amount.

\begin{lemma}\label{lemma:clusterdecrease}
    There exists a deterministic $O(d+1)$-round \CONGEST algorithm that, given an edge-weighted tree $T=(V,E,w)$ with a $(d,\lambda)$-clustering and a parameter $0<\delta\leq1$, computes a $(\lfloor 2/\delta\rfloor, \min\{1,\lambda+\delta\})$-clustering of $T$.
\end{lemma}
\begin{proof}
    We focus on one cluster $V_i$. Singleton clusters remain unchanged. If its diameter is at most $2\lfloor1/\delta\rfloor$, no splitting is needed. Otherwise, choose a root $v_i\in V_i$ \reviewnote{Explain root selection.}\reviewchange{using standard leader election} and compute every node's distance from $v_i$ in $O(d+1)$ rounds. Let $q:=\lceil1/\delta\rceil$ and let $V_{i,r}$ contain the nodes at distance $r$. Every edge of $T[V_i]$ joins consecutive layers. For each $p\in\{0,\dots,q-1\}$, let $E_p$ contain the edges between layers $r$ and $r+1$ with $r\equiv p\pmod q$. These sets partition the edges of $T[V_i]$. In the branch under consideration, $\operatorname{diam}(T[V_i])>2\lfloor1/\delta\rfloor$, so $q=O(d+1)$. Aggregate the $q$ sums $\sum_{e\in E_p}w(e)$ in a fixed order and pipeline them toward $v_i$; this takes $O(d+q)=O(d+1)$ rounds. The root chooses a minimum-weight residue $p_0$ and broadcasts its index. Its weight is at most $1/q\leq\delta$ times the weight of $T[V_i]$. Removing $E_{p_0}$ creates clusters of diameter at most $2(q-1)\reviewchange{\leq 2\lfloor1/\delta\rfloor}\leq\lfloor2/\delta\rfloor$. Summed over all original clusters, the added intercluster-edge weight is at most $\delta$ times the total intracluster-edge weight and hence at most $\delta w(E)$. Since the total intercluster-edge weight is always at most $w(E)$, the output parameter can be capped at $1$.
\end{proof}

We now have everything we need to prove our main clustering result for edge-weighted trees, which is given by the following theorem.

\begin{theorem}\label{thm:clustering}
 There exists a deterministic $O\big(\frac{\log^* n}{\delta}\big)$-round \CONGEST algorithm that, for any parameter $0<\delta\leq1$, computes a $\big(O(1/\delta), \delta\big)$-clustering of any edge-weighted tree $T=(V,E,w)$.
\end{theorem}
\begin{proof}
  We maintain the following invariant for a current parameter $\lambda\in(0,1]$: the total intercluster-edge weight is at most $\lambda w(E)$ and every cluster has diameter at most $4/\lambda$. Suppose such a clustering is given. Applying \Cref{lemma:clusterincrease} five times reduces the intercluster-edge parameter from $\lambda$ to $(3/4)^5\lambda<\lambda/4$. \reviewnote{Give explicit diameter bound.}\reviewchange{After these five applications, the cluster diameter is at most $9^5(4/\lambda+1)-1=O(1/\lambda)$.} The five merge steps take $O(\log^*n/\lambda)$ rounds.

  We next apply \Cref{lemma:clusterdecrease} with parameter $\lambda/4$. This produces a $(\lfloor8/\lambda\rfloor,\lambda/2)$-clustering in $O(1/\lambda)$ rounds: the previously remaining intercluster edges contribute less than $\lambda w(E)/4$, and splitting contributes at most another $\lambda w(E)/4$. Moreover, $\lfloor8/\lambda\rfloor\leq4/(\lambda/2)$, so the diameter part of the invariant is restored for the new parameter $\lambda/2$. Thus, one phase turns a clustering satisfying the invariant for $\lambda$ into one satisfying it for $\lambda/2$ in $O(\log^*n/\lambda)$ rounds.

  The singleton clustering satisfies the invariant for $\lambda_0=1$. Let $\lambda_i:=2^{-i}$ and repeat the phase until $L:=\lceil\log(1/\delta)\rceil$. Because $\delta/2<\lambda_L\leq\delta$, the final clustering has diameter $O(1/\delta)$ and intercluster-edge weight at most $\delta w(E)$. The total round complexity is the geometric sum $O\left(\sum_{i=0}^L2^i\log^*n\right)=O(\log^*n/\delta)$.
\end{proof}

\subsection{MWIS Approximation Algorithm}

We can now provide the details of our distributed approximation algorithm for the MWIS problem in trees. As described in \Cref{sec:contribTrees}, the algorithm is based on ideas developed in \cite{CzygrinowHW08}. The high-level idea of the algorithm of \cite{CzygrinowHW08} is the following. Given a tree $T=(V,E)$ with node weights $w:V\to\mathbb{Z}_{\geq0}$, one first defines edge weights $w:E\to\mathbb{Z}_{\geq0}$ as follows. For every edge $e=\set{u,v}\in E$, the weight of $e$ is defined as
\begin{equation}\label{eq:edgeweight}
    \forall e=\set{u,v}\in E\,:\,
    w(e) := \min\set{w(u), w(v)}.
\end{equation}
The algorithm now consists of three main steps:

\begin{enumerate}
    \item For a given parameter $0<\eps\leq 1/2$, compute an $\big(O(1/\eps), \eps/2\big)$-clustering of the weighted tree $T$. Let $V_1,\dots,V_s\subseteq V$ be the resulting clusters.
    \item Compute an optimal (weighted) independent set $I_i\subseteq V_i$ for the induced subtree $T[V_i]$ for every cluster $V_i$. Let $J:=I_1\cup I_2\cup \dots\cup I_s$.
    \item For every intercluster edge $\set{u,v}$ (i.e., for every edge $\set{u,v}$ such that $u$ and $v$ are in different clusters), if both $u$ and $v$ are in $J$, remove the lower-weight endpoint from $J$, breaking ties by ID. Let $I$ be the resulting set.
\end{enumerate}

The algorithm clearly computes an independent set. The following lemma shows that the independent set is close to optimal if the total weight of all intercluster edges is small.

\begin{lemma}\cite{CzygrinowHW08}\label{lemma:clusteringIS}
    Let $I$ be the independent set computed by the above algorithm and let $I^*$ be an optimal independent set. We have $w(I) \geq (1-\eps)\cdot w(I^*)$.
\end{lemma}
\begin{proof}
    Let $E_{\mathrm{cut}}$ be the set of intercluster edges. Before Step 3, we have $w(J)\geq w(I^*)$. After Step 3, we therefore have $w(I) \geq w(I^*) - \sum_{e\in E_{\mathrm{cut}}} w(e) \geq w(I^*) - \frac{\eps}{2}\cdot w(E)$.
    To prove the lemma, first observe that $w(E) \leq w(V)$. To see this, define an arbitrary node $r\in V$ as root and consider $T$ as a rooted tree. For each node $v\in V\setminus\set{r}$, the weight of the edge connecting $v$ to its parent is at most $w(v)$. The total weight of all edges is therefore at most $\sum_{v\in V\setminus\set{r}} w(v) \leq w(V)$. Because a tree is two-colorable, we further know that $w(I^*)\geq w(V)/2$. We therefore have $\frac{\eps}{2}\cdot w(E) \leq \frac{\eps}{2}\cdot w(V)\leq \eps\cdot w(I^*)$.
\end{proof}

Computing an optimal independent set $I_i$ within some cluster $V_i$ can be done in time linear in the diameter of $T[V_i]$ by applying a standard dynamic programming algorithm. The proof of the following lemma appears below.

\begin{lemma}\label{lemma:treeOpt}
	There exists an $O(D+1)$-round deterministic \CONGEST algorithm that, for any tree $T=(V,E)$ with node weights $w:V\to \mathbb{Z}_{\geq0}$ and diameter $D$, optimally solves the MWIS problem on $T$.
\end{lemma}

\begin{proof} 
    A singleton tree is solved locally by selecting its only node. For every other tree, synchronous leaf peeling finds one or two centers in $O(D+1)$ rounds. In the latter case, the center with smaller ID becomes the root. A broadcast from the root then orients all edges toward it. We apply the standard tree dynamic program to this rooted tree.

    In a bottom-up pass, every node $u$ computes the optimum values $\OPT_{in}(u)$ and $\OPT_{out}(u)$ for its subtree, subject to $u$ being in or out of the independent set. Define $\OPT(u):=\max\{\OPT_{in}(u),\OPT_{out}(u)\}$. If $N_c(v)$ is the set of children of $v$, then
    \begin{eqnarray*}
        \OPT_{in}(v) & := & w(v) + \sum_{u \in N_{c}(v)} \OPT_{out}(u),\\
        \OPT_{out}(v) & := & \sum_{u \in N_{c}(v)} \OPT(u).
    \end{eqnarray*}
    Which nodes are in the independent set is then computed in a top-down fashion. For the root node $r$, $r$ is added to the independent set iff $\OPT(r)=\OPT_{in}(r)$. For all $v \in V \setminus \{r\}$, $v$ is added to the independent set iff its parent is not in the independent set and $\OPT(v)=\OPT_{in}(v)$.

    Each value is at most $nB$ under the weight convention of \Cref{sec:prelim}, so it fits in $O(\log n)$ bits. The two passes, including the singleton case, take $O(D+1)$ rounds.
\end{proof}

We can now prove the upper bound part of \Cref{thm:mainTree}.

\begin{lemma}\label{lemma:treeUpper}
 There exists a deterministic $O\big(\frac{\log^* n}{\eps}\big)$-round \CONGEST algorithm that, for any parameter $0<\eps\leq 1/2$, computes a $(1-\eps)$-approximation to MWIS in trees.
\end{lemma}
\begin{proof}
  In \Cref{lemma:clusteringIS}, we have already established that the algorithm described above computes a $(1-\eps)$-approximation of MWIS in trees. It remains to show that the algorithm can be implemented in $O\big(\frac{\log^* n}{\eps}\big)$ rounds in the \CONGEST model. In the first step of the algorithm, we compute an $\big(O(1/\eps), \eps/2\big)$-clustering of the tree. By \Cref{thm:clustering}, this can be done in $O\big(\frac{\log^* n}{\eps}\big)$ rounds in \CONGEST. The second step computes an optimal independent set in each cluster. Because each cluster has diameter $O(1/\eps)$, this can be done in $O(1/\eps)$ rounds by \Cref{lemma:treeOpt}. The last step of the algorithm only needs one additional round.
\end{proof}

\subsection{Approximation Algorithms for Other Problems}

We next show that with the clustering provided by \Cref{thm:clustering}, we can also compute $(1\pm\eps)$-approximations for the unweighted minimum dominating set (MDS), minimum vertex cover (MVC) and maximum matching (MM) problems.

\begin{theorem}\label{thm:upperOther}
For every $0<\eps\leq 1/2$, there exist deterministic $O\big(\frac{\log^* n}{\eps}\big)$-round \CONGEST algorithms on trees that compute $(1+\eps)$-approximations for unweighted MDS and MVC and a $(1-\eps)$-approximation for unweighted MM.
\end{theorem}
\begin{proof}[Proof sketch]
    Trees with at most two nodes are solved directly. Otherwise, let $m$ be the number of non-leaf nodes. For the analysis, root the tree at a non-leaf. \reviewnote{Justify the disjoint paths.}\reviewchange{Each non-leaf chooses exactly one child, and each vertex can be chosen only by its parent. Thus the $m$ chosen edges form a subgraph of maximum degree at most two. Since the tree is acyclic, these edges form vertex-disjoint paths. Taking alternate edges on each path gives a matching of size at least $m/2$. By K\H{o}nig's theorem~\cite{koenig31}, the minimum vertex cover also has size at least $m/2$.} The $m$ chosen children intersect every closed neighborhood in at most three nodes, so every dominating set has size at least $m/3$. Thus, all three optima are at least $m/3$. Conversely, the non-leaf nodes form a dominating set and a vertex cover, and every matching edge can be assigned to a distinct non-leaf, so all three optima are at most $m$. Apply \Cref{thm:clustering} to the subtree induced by the non-leaf nodes with parameter $\eps/3$, then attach each leaf to its neighbor's cluster. The clusters have diameter $O(1/\eps)$ and at most $\eps m/3$ intercluster edges. Solve each cluster optimally by tree dynamic programming. For MDS and MM, the union of the cluster solutions is feasible; for MVC, add one endpoint of each uncovered intercluster edge. In each case, the additive error is at most the number of intercluster edges, giving the claimed approximation and round complexity. The full proof appears in Appendix~\ref{app:upperOther}.
\end{proof}

Finally, we prove the matching lower bound for \Cref{thm:mainTree}.

\medskip

\subsection{Lower Bound}  The starting point of the lower bound proof is the fact that any deterministic algorithm computing a constant-factor approximation to the MCIS problem in rings requires $\Omega(\log^* n)$ rounds. This was proven independently in \cite{lenzen08} and \cite{CzygrinowHW08} using different methods. The Ramsey argument of \cite{CzygrinowHW08} remains valid when the ring is given a consistent orientation. For completeness, we next transfer this lower bound to consistently oriented paths while preserving the exact size parameter known to the nodes. The proof appears below.

\begin{lemma}\label{lemma:orientedpaths_constant}
    For every constant $c\geq1$, there are constants $a_c>0$ and $N_c$ such that, for every $N\geq N_c$, every deterministic \LOCAL algorithm that returns an independent set of size at least $\alpha(P)/c$ on every consistently oriented $N$-node path has worst-case round complexity at least $a_c\log^*N$. This holds even if nodes know the exact value of $N$ and fixed additional global parameters independent of the identifier assignment.
\end{lemma}

\begin{proof}
    Fix a constant $c\geq1$. The Ramsey lower bound for cycles in \cite{CzygrinowHW08} has the following quantitative form for the fixed approximation factor $2c$: there are constants $b>0$ and $M$ such that, for every $N\geq M$, every deterministic \LOCAL algorithm computing a $2c$-approximation on consistently oriented $N$-node cycles has worst-case round complexity at least $b\log^*N$. This remains true when the nodes know the exact value of $N$ and fixed additional global parameters. For a $t$-round algorithm, fix the orientation and place each increasing $(2t+1)$-tuple of identifiers consecutively in the orientation direction; the central output defines the same coloring of identifier subsets as in the cited Ramsey argument.

    Set $a_c:=b/2$ and increase $N_c\geq M$ so that $2(a_c\log^*N+1)\leq N/(4c)$ for every $N\geq N_c$. Consider any such $N$ and suppose that an algorithm $\mathcal A$ computes a $c$-approximation on every consistently oriented $N$-node path in $t<a_c\log^*N$ rounds. In particular, it always returns an independent set of size at least $\alpha(P)/c\geq N/(2c)$. We define a $t$-round algorithm $\mathcal B$ on consistently oriented $N$-node cycles. Each node applies the local decision rule of $\mathcal A$ for an internal path node to its radius-$t$ oriented neighborhood and supplies the same exact value $N$ and the same additional global parameters to that rule.

    The output of $\mathcal B$ is independent. For either endpoint of any cycle edge, the two relevant radius-$t$ views can be realized simultaneously away from the endpoints of an oriented $N$-node path, with the same identifiers; this is possible because $2t+2<N$. The corresponding decisions of $\mathcal A$ therefore cannot both select the endpoints of the edge. To bound the output size, delete an arbitrary edge of a labeled oriented cycle $C$ and let $P$ be the resulting oriented path, retaining the same identifiers and supplied parameters. The executions of $\mathcal A$ on $P$ and $\mathcal B$ on $C$ differ only at nodes within distance $t$ of the deleted edge. Consequently,
    \[
        |\mathcal B(C)|\geq |\mathcal A(P)|-2(t+1)
        \geq \frac{N}{2c}-2(t+1)\geq\frac{N}{4c}
        \geq\frac{\alpha(C)}{2c}.
    \]
    Thus $\mathcal B$ is a $2c$-approximation running in fewer than $b\log^*N$ rounds, contradicting the quantitative cycle lower bound. Hence $t\geq a_c\log^*N$, as claimed.
\end{proof}

We can now prove the lower bound for computing $(1\pm\eps)$-approximations of maximum cardinality independent set (MCIS), maximum matching (MM), minimum vertex cover (MVC), and minimum dominating set (MDS) in oriented paths. The proof of the following lemma appears below.

\begin{lemma}\label{lemma:orientedpathslower}
    Every deterministic \LOCAL algorithm to compute a $(1\pm\eps)$-approximation, where  $\log(n)/n<\eps\leq1/8$ is any parameter, for MCIS, MM, MVC, or MDS in oriented paths requires  $\Omega\big(\frac{\log^* n}{\eps}\big)$ rounds.
\end{lemma}

\begin{proof}
    Fix $\log(n)/n<\eps\leq1/8$. It suffices to consider sufficiently large $n$. We first prove the lower bound for MCIS, then reduce MCIS to MVC and MM, and finally give a separate reduction for MDS.

    \subsubsection{MCIS}
    Suppose that an algorithm $\calA$ computes a $(1-\eps)$-approximate independent set in oriented $n$-node paths in $T(n,\eps)$ rounds. Define
    \[
        k:=\left\lfloor\frac{1}{4\eps}-\frac12\right\rfloor\geq1,
        \qquad L:=2k+1,
        \qquad N:=\left\lceil\frac nL\right\rceil.
    \]
    Here $L=\Theta(1/\eps)$ and $L\leq1/(2\eps)$. Given an oriented $N$-node path $P$, replace each edge by an oriented path of length $L$. The resulting path has $(N-1)L+1$ nodes. Add an oriented prefix of
    \[
        r:=n-\big((N-1)L+1\big),\qquad 0\leq r\leq L-1,
    \]
    nodes to obtain a path $P'$ with exactly $n$ nodes. We call the nodes of $P$ the \emph{supernodes} of $P'$; all other nodes are virtual.

    To simulate $P'$, each physical node of $P$ owns its supernode and the $L-1$ virtual nodes on its outgoing edge, if it has one. The first physical node also owns the prefix. These are consecutive blocks, all of length $L$ except for the first and last blocks, whose lengths are $L+r$ and $1$, respectively. Each simulated node receives a unique identifier consisting of its owner's identifier, a type tag, and a local offset. Since the original identifiers use $O(\log N)$ bits and the offsets are at most $O(n)$, these identifiers use $O(\log n)$ bits. The simulated nodes are given $n$ and $\eps$ as global parameters. In the \LOCAL model, each physical node can gather the descriptions of all blocks meeting the radius-$t$ neighborhoods of its simulated nodes in $O(1+t/L)$ rounds, then simulate their outputs locally. Thus, running $\calA$ on $P'$ takes $O(1+\eps T(n,\eps))$ rounds on $P$.

    Let $S$ be the independent set returned by $\calA$. We transform $S$ into an independent set $S'$ of at least the same size with no pair of consecutive supernodes selected. For each edge $(u,v)$ of $P$ whose two supernodes both belong to the original set $S$, remove $u$ and replace the entire selection among the $2k$ internal nodes of the corresponding path in $P'$ by the nodes in positions $1,3,\dots,2k-1$, measured from $u$. Originally there were at most $k-1$ selected internal nodes, so this replacement does not decrease the total size, including the removed $u$.

    All replacements use the original set $S$ and are performed simultaneously. Their interiors are disjoint, and each removed supernode is charged only to its outgoing interval. The new selection inside each replaced interval is independent. Among its two endpoints, only $u$ is adjacent to a newly selected internal node, and $u$ is removed. Hence the replacements are compatible with selections in adjacent intervals and in the prefix. The result $S'$ is independent, $|S'|\geq|S|$, and no two consecutive supernodes remain selected. The physical nodes implement these replacements with $O(1)$ additional rounds.

    The selected supernodes therefore form an independent set $J$ of $P$. There are at most $k(N-1)$ selected internal nodes and at most $\lceil r/2\rceil$ selected prefix nodes. Since $n-N=2k(N-1)+r$ and $|S'|\geq(1-\eps)\lceil n/2\rceil$, we obtain
    \[
        \begin{aligned}
            |J|&\geq (1-\eps)\frac n2-k(N-1)-\left\lceil\frac r2\right\rceil\\
                &\geq\frac N2-\frac{\eps n+1}{2}
                 \geq\frac N4-\frac12.
        \end{aligned}
    \]
    The last inequality uses $n\leq NL$ and $\eps L\leq1/2$. Thus $J$ is a fixed constant-factor approximation for sufficiently large $N$. By the quantitative statement of \Cref{lemma:orientedpaths_constant}, there is a constant $a>0$, independent of $n$ and $\eps$, such that computing $J$ requires at least $a\log^*N$ rounds. The simulation and repair take at most $C(1+\eps T(n,\eps))$ rounds for an absolute constant $C$. Since $N=\Theta(\eps n)$ and $\eps n>\log n$, we have $\log^*N=\Omega(\log^*n)$, and therefore
    \[
        C(1+\eps T(n,\eps))\geq a\log^*N
        =\Omega(\log^*n).
    \]
    For sufficiently large $n$, the additive constant can be absorbed, yielding $T(n,\eps)=\Omega(\log^*n/\eps)$.
    In fact, this argument only requires $\eps n=\Omega(\log n)$, together with $0<\eps\leq1/8$; we will use this observation to handle the one-node size difference in the matching reduction.

    \subsubsection{MVC}
    Let $C$ be a $(1+\eps)$-approximate vertex cover of an oriented $n$-node path. Its complement $S=V\setminus C$ is independent and satisfies
    \[
        \begin{aligned}
            |S|&\geq n-(1+\eps)\lfloor n/2\rfloor\\
                &=\lceil n/2\rceil-\eps\lfloor n/2\rfloor
                 \geq(1-\eps)\lceil n/2\rceil.
        \end{aligned}
    \]
    Thus an approximate MVC algorithm gives a $(1-\eps)$-approximate MCIS algorithm on the same path with no communication overhead. The MCIS lower bound therefore also holds for MVC.

    \subsubsection{MM}
    A matching in a path $Q$ is exactly an independent set in its line graph, which is a path with one fewer node. To use this correspondence in the required direction, let $P$ be an oriented path with $m:=n-1$ nodes. Denote its nodes in path order by $v_1,\dots,v_m$ and construct a virtual oriented path $Q$ with nodes $x_0,\dots,x_m$, where the edge $e_i=(x_{i-1},x_i)$ corresponds to $v_i$. Physical node $v_i$ simulates $x_i$, and $v_1$ also simulates $x_0$. Identifiers are formed from the owner's identifier and a tag, so they are unique and use $O(\log n)$ bits. The indices here only describe the construction; the consistent orientation identifies each node's predecessor and successor, and no node needs to know its position.

    Each round on $Q$ can be simulated in $O(1)$ rounds on $P$. From a matching $M$ in $Q$, select $v_i$ exactly when $e_i\in M$, using at most one further round to communicate edge membership. This gives an independent set of $P$ of size $|M|$. The maximum matching size of $Q$ is $\lfloor n/2\rfloor=\lceil m/2\rceil=\alpha(P)$, so a $(1-\eps)$-approximate matching gives a $(1-\eps)$-approximate MCIS solution. Since $\eps m=\Omega(\log m)$ follows from $\eps n>\log n$, the MCIS lower bound proved above applies and yields the claimed $\Omega(\log^*n/\eps)$ bound for MM.

    \subsubsection{MDS}
    Suppose that an algorithm computes a $(1+\eps)$-approximate dominating set on oriented $n$-node paths in $T(n,\eps)$ rounds. Set
    \[
        k:=\left\lfloor\frac{1}{6\eps}\right\rfloor\geq1,
        \qquad L:=3k-1,
        \qquad N:=\left\lfloor\frac nL\right\rfloor,
        \qquad r:=n-NL.
    \]
    We have $L=\Theta(1/\eps)$ and $0\leq r<L$. Given an oriented $N$-node path $P$, replace each physical node by a consecutive block of $L$ virtual nodes, join consecutive blocks in path order, and append a tail of $r$ nodes to the last block. This gives a virtual path $P'$ with exactly $n$ nodes. Each physical node simulates its block, and the last node also simulates the tail. Using the identifier encoding and neighborhood simulation described for MCIS, a $T(n,\eps)$-round algorithm on $P'$ can be simulated in $O(1+\eps T(n,\eps))$ rounds on $P$.

    Let $D$ be the resulting $(1+\eps)$-approximate dominating set of $P'$. Every full block contains at least $k-1$ members of $D$: nodes outside the block can dominate at most its two endpoints, and each selected node inside the block dominates at most three block nodes. Call a full block \emph{cheap} if it contains exactly $k-1$ members of $D$. For such a block the coverage bound is tight, since
    \[
        3(k-1)+2=3k-1=L.
    \]
    Consequently, both nodes immediately outside a cheap block must exist and belong to $D$. Neither endpoint inside the block can belong to $D$: a selected endpoint would dominate at most two block nodes, contradicting the tight coverage bound. Two adjacent full blocks cannot both be cheap, because a cheap left block requires the first node in the right block to be selected, whereas a cheap right block requires that node to be unselected. Thus the physical nodes whose blocks are cheap form an independent set $J$ in $P$.

    Write $Q:=|J|$, and partition $D$ into the selected nodes $D_F$ in the full blocks and $D_T$ in the tail. Every non-cheap full block contains at least $k$ selected nodes, so $|D_F|\geq kN-Q$. At most one tail node can be dominated from outside the tail; hence $|D_T|\geq(r-1)/3$, which also holds when $r=0$. Since the domination number of an $n$-node path is $\lceil n/3\rceil$, we have
    \[
        \begin{aligned}
            Q&\geq kN-|D_F|\\
             &\geq kN-(1+\eps)\left\lceil\frac{NL+r}{3}\right\rceil
                    +\frac{r-1}{3}\\
             &\geq\left(\frac{1+\eps}{3}-\eps k\right)N
                    -\frac{\eps r+3+2\eps}{3}\\
             &\geq\frac N6-2.
        \end{aligned}
    \]
    The third inequality uses $\lceil n/3\rceil\leq(n+2)/3$ and $L=3k-1$. The last uses $\eps k\leq1/6$, $r<3k$, and $\eps\leq1/8$. Therefore $J$ is a fixed constant-factor approximation for sufficiently large $N$. A physical node can determine whether its block is cheap locally after the simulation. As above, there are constants $a,C>0$, independent of $n$ and $\eps$, such that \Cref{lemma:orientedpaths_constant} gives
    \[
        C(1+\eps T(n,\eps))\geq a\log^*N=\Omega(\log^*n),
    \]
    where $N=\Theta(\eps n)$ and $\eps n>\log n$. Absorbing the additive constant for sufficiently large $n$ gives $T(n,\eps)=\Omega(\log^*n/\eps)$, as required.
\end{proof}

\reviewnote{Answer high-girth question.}\reviewchange{The MWIS upper bound also extends to high-girth graphs in proper minor-closed families.\footnote{\reviewchange{For each fixed proper minor-closed family, girth at least $C/\eps$ for a sufficiently large family-dependent constant $C$ keeps all intermediate clusters acyclic and leaves at most one edge between any two clusters, avoiding congestion from aggregating parallel edges. The contracted graphs remain in the family and retain bounded arboricity, so the merge step still absorbs a constant fraction of the intercluster edge weight. See the full version for details.}}}

\section{Approximating MWIS as a Function of the Arboricity}
\label{sec:arboricity}

We now discuss our algorithms for approximating MWIS as a function of the arboricity.

\subsection{Simple MWIS Approximation by Using Local Rounding}
\label{sec:arboricity_rounding}
In this section, we warm up by proving \Cref{thm:arboricity_rounding}, which follows by a simple application of the local rounding technique of \cite{Localrounding23}.

\begin{proof}[Proof sketch of \Cref{thm:arboricity_rounding}]
    For $\vec{x}\in[0,1]^V$, define
    \[
        \utility(\vec{x}):=\sum_{v\in V}w(v)x_v
        \quad\text{and}\quad
        \cost(\vec{x}):=\sum_{\{u,v\}\in E}\min\{w(u),w(v)\}x_ux_v.
    \]
    Set $x_v:=1/(2\beta)$. An orientation with outdegree at most $\beta$, used only for the analysis, gives $\utility(\vec{x})=w(V)/(2\beta)$ and $\cost(\vec{x})\leq w(V)/(4\beta)$. For this binary labeling, both label probabilities are at least $1/(2\beta)$. Applying the generic rounding lemma of \cite{Localrounding23} with $\mu=1/2$ and $\lambda_{\min}=1/(2\beta)$ yields an integral vector $\vec{y}$ preserving at least a $(1-\eps)$-fraction of $\utility(\vec{x})-\cost(\vec{x})$. Among the selected nodes, simultaneously remove the lighter endpoint of each edge, breaking ties by ID. The removed weight is at most $\cost(\vec{y})$, so the resulting independent set satisfies $w(I)\geq\utility(\vec{y})-\cost(\vec{y})\geq(1-\eps)w(V)/(4\beta)$. The full proof, including the rounding parameters and \CONGEST implementation, appears below.
\end{proof}

\begin{proof}
    We use the generic rounding framework of \cite{Localrounding23}. For a fractional selection vector $\vec{x}\in[0,1]^V$, define
    \[
        \utility(\vec{x}):=\sum_{v\in V}w(v)x_v,
        \qquad
        \cost(\vec{x}):=\sum_{\{u,v\}\in E}\min\{w(u),w(v)\}x_ux_v.
    \]
    This is a binary labeling problem: node $v$ has label $1$ with probability $x_v$ and label $0$ with probability $1-x_v$. Its utility is $w(v)$ for label $1$ and zero for label $0$. An edge $\{u,v\}$ has cost $\min\{w(u),w(v)\}$ when both endpoints have label $1$, and zero otherwise. All these utilities and costs are nonnegative. Exchanging node weights and IDs in $O(1)$ rounds makes each edge's cost function known to both endpoints.

    Set $x_v:=1/(2\beta)$ for all $v$. For the analysis, fix an orientation of $G$ with outdegree at most $\beta$, which exists by orienting each tree in a partition into $\beta$ forests towards a root. The algorithm does not need to compute this orientation. Let $N_{\mathrm{out}}(v)$ denote the outneighbors of $v$. Then
    \[
        \begin{aligned}
            \utility(\vec{x})&=\frac{w(V)}{2\beta},\\
            \cost(\vec{x})
                &=\frac{1}{4\beta^2}\sum_{v\in V}\sum_{u\in N_{\mathrm{out}}(v)}\min\{w(u),w(v)\}\\
                &\leq\frac{1}{4\beta^2}\sum_{v\in V}|N_{\mathrm{out}}(v)|w(v)
                 \leq\frac{w(V)}{4\beta}.
        \end{aligned}
    \]
    Thus $\utility(\vec{x})-\cost(\vec{x})\geq\utility(\vec{x})/2$. Moreover, since $\beta\geq1$, both label probabilities $x_v=1/(2\beta)$ and $1-x_v\geq1/2$ are at least $\lambda_{\min}:=1/(2\beta)$. The unique IDs give an initial proper $\zeta$-coloring with $\zeta=n^{O(1)}$.

    We apply the generic rounding lemma~\cite[Lemma~2.5]{Localrounding23} with accuracy $\eps$, $\mu=1/2$, and $\lambda_{\min}=1/(2\beta)$, incorporating node utilities as in Lemma~2.4 of the same paper. It computes an integral vector $\vec{y}\in\{0,1\}^V$ satisfying
    \[
        \utility(\vec{y})-\cost(\vec{y})
        \geq(1-\eps)\bigl(\utility(\vec{x})-\cost(\vec{x})\bigr)
        \geq(1-\eps)\frac{w(V)}{4\beta}.
    \]
    Since $1/(\eps\mu\lambda_{\min})=4\beta/\eps$, the round complexity is
    \[
        O\left(\frac{\log^2(4\beta/\eps)}{\eps}
            +\log(4\beta/\eps)\cdot\log^*\zeta\right)
        =O\left(\frac{\log^2(\beta/\eps)}{\eps}
            +\log(\beta/\eps)\cdot\log^*n\right).
    \]
    The lemma uses $O(\log n)$-bit messages because the label alphabet has size two and $\log\zeta=O(\log n)$. Its node-utility extension is simulated locally at each node and adds no communication overhead.

    Finally, let $S:=\{v:y_v=1\}$. For every edge with both endpoints in $S$, mark its lower-weight endpoint for removal, breaking equal-weight ties by marking the endpoint with smaller ID. Remove all marked nodes simultaneously. This yields an independent set $I$. The total removed weight is at most
    $\sum_{\{u,v\}\in E(G[S])}\min\{w(u),w(v)\}=\cost(\vec{y})$,
    so $w(I)\geq\utility(\vec{y})-\cost(\vec{y})$. This final step takes $O(1)$ rounds and incurs no additional approximation loss, proving both claims of the theorem.
\end{proof}

\subsection{Iterative Computation of Weighted Independent Sets}

\begin{algorithm}[t]
    \caption{Independent set composition for graph $G=(V,E)$ with node weights $w:V\to \mathbb{Z}_{\geq0}$} \label{alg:localratio}
    \begin{algorithmic}[1]
        \State Let $I_0\subseteq V$ be an arbitrary independent set of $G$
        \State Define node weights $w':V\to\mathbb{Z}_{\geq0}$ as
        \[
        w'(v) := \max\Bigg\{0, w(v) - \sum_{u\in I_0\cap N^+(v)}w(u)\Bigg\}
        \]
        \State Let $I'$ be an independent set of $G$ such that $w'(v)>0$ for all $v\in I'$
        \State \reviewnote{Disjoint sets use $N$.} \textbf{return} $I:= I' \cup \big(I_0 \setminus \reviewchange{N(I')}\big)$
    \end{algorithmic}
\end{algorithm}

Before we can prove \Cref{thm:carowei} in the upcoming section, we revisit a simple method to iteratively approximate the MWIS problem. The method is based on the local-ratio technique of \cite{localratio} and it has been explicitly described and used in the context of distributed algorithms for MWIS in \cite{Bar-YehudaCGS17,KawarabayashiKS20,Localrounding23}. \Cref{alg:localratio} gives a generic way to iteratively compose independent sets in a meaningful way in the weighted setting. Note that Line 3 of \Cref{alg:localratio} in particular implies that $I_0$ and $I'$ are disjoint because for all $v\in I_0$, we have $w'(v)=0$. The following lemma proves that \Cref{alg:localratio} returns an independent set of weight at least $w(I_0) + w'(I')$.

\begin{lemma}\label{lemma:localratio}
    When applying \Cref{alg:localratio} to a graph $G=(V,E)$ with node weights $w:V\to \mathbb{Z}_{\geq0}$, the weight of the returned independent set $I$ satisfies $w(I) \geq w(I_0) + w'(I')$.
\end{lemma}
\begin{proof}
    Consider some node $v\in I'$. Note that because we have $w'(v)>0$, we have $w'(v) = w(v) - \sum_{u\in I_0\cap N^+(v)} w(u)$. We therefore have\reviewnote{Match the return rule.}
    \begin{eqnarray*}
        w(I) & = &
        \sum_{v\in I_0} w(v) + \sum_{v\in I'} w(v) - \sum_{v\in I_0\cap \reviewchange{N(I')}} w(v)\\
        & \geq &
        \sum_{v\in I_0} w(v) + \sum_{v\in I'} w(v) - \sum_{v\in I'}\sum_{u\in N^+(v)\cap I_0} w(u)\\
        & = &
        \sum_{v\in I_0} w(v) +
        \sum_{v\in I'} \left[w(v) - \sum_{u\in I_0\cap N^+(v)} w(u)\right]\\
        & = & 
        w(I_0) + w'(I').
    \end{eqnarray*}
\end{proof}

We can use \Cref{lemma:localratio} to compute a sequence of independent sets $I_0, I_1, \dots, I_T$ that we can then combine to a single independent set. The following lemma states this explicitly; its proof appears below.

\begin{lemma}\label{lemma:localratioseq}
    Let $G=(V,E)$ be a graph with node weights $w:V\to\mathbb{Z}_{\geq0}$ and for some integer $T\geq 0$, let $I_0, I_1, \dots,I_T$ be a sequence of $T+1$ independent sets and $w_0,w_1,\dots,w_T$ be a sequence of node weight functions that satisfy the following conditions. For all nodes $v\in V$, $w_0(v):=w(v)$. Further, for any $i\geq 1$ and all $v\in V$, we have $w_i(v) := \max\Bigg\{0, w_{i-1}(v) - \sum_{u\in N^+(v)\cap I_{i-1}} w_{i-1}(u) \Bigg\}$, and the independent set $I_i$ only consists of nodes $u$ for which $w_i(u)>0$. There exists a $T$-round deterministic \CONGEST algorithm that, for independent sets $I_0, I_1, \dots,I_T$, computes an independent set $I$ of weight $w(I) \geq \sum_{i=0}^T w_i(I_i)$.
\end{lemma}

\begin{proof}
    Define $\bar{I}_T:=I_T$ and, for each $i\in\{0,\dots,T-1\}$, let $\bar{I}_i:=\bar{I}_{i+1}\cup\bigl(I_i\setminus N^+(\bar{I}_{i+1})\bigr)$. Residual weights are pointwise nonincreasing. Moreover, every $v\in\bar I_{i+1}$ belongs to some $I_j$ with $j\geq i+1$, so $w_{i+1}(v)\geq w_j(v)>0$. Thus the positivity hypothesis of \Cref{lemma:localratio} holds, and for each $i\in\{0,\dots,T-1\}$ that lemma gives
    \[
    w_i(\bar{I}_i) \geq w_i(I_i) + w_{i+1}(\bar{I}_{i+1}).
    \]
    Applying this inductively, we thus get
    \[
    w_i(\bar{I}_i) \geq \sum_{j=i}^T w_j(I_j).
    \]
    Note that we have $w_T(\bar{I}_T)=w_T(I_T)$ because $\bar{I}_T=I_T$. By setting $I:=\bar{I}_0$, we thus have
    \[
    w(I) = w(\bar{I}_0) \geq
    \sum_{j=0}^T w_j(I_j),
    \]
    as required. Clearly, given $\bar{I}_{i+1}$, the set $\bar{I}_{i}$ can be computed in a single round of the \CONGEST model. The independent set $I$ can therefore be computed in $T$ rounds.
\end{proof}

\subsection{Approximating a Weighted Caro-Wei Bound}

We now have everything that we need to describe and analyze our algorithm to prove \Cref{thm:carowei} and thus approximate the weighted generalization of the Caro-Wei bound \eqref{eq:carowei}. The following statement is proven in \cite{Localrounding23}.

\begin{lemma}\label{lemma:Deltaapprox}\emph{(Theorem 1.3 of~\cite{Localrounding23})}
    There exists a deterministic $O(\log^2\Delta\cdot\log(1/\eps) + \log^* q)$-round \CONGEST algorithm that, for every graph $G=(V,E)$ with node weight function $w:V\to\mathbb{Z}_{\geq0}$, maximum degree $\Delta$, an initial proper $q$-coloring with $q\leq n^{O(1)}$,\footnote{This way each color has an $O(\log n)$-bit encoding.} and for any $0<\eps\leq 1/2$, computes an independent set of $G$ of weight at least $(1-\eps)\cdot w(V)/(\Delta+1)$.
\end{lemma}
The cited theorem is stated for positive weights; zero-weight nodes can be omitted from its invocation without changing the total weight or increasing the maximum degree.

Given the same prerequisites as in \Cref{lemma:Deltaapprox}, we next describe an algorithm to obtain an independent set satisfying the conditions of \Cref{thm:carowei}. If $\Delta=0$, return $V$, which is an optimal independent set. Otherwise, isolated nodes are included locally as the class $V_0$ below. First, we define $\eps':= \eps/2$ and note that we will be working with $\eps'$ instead of $\eps$ for the rest of the algorithm. For $\Delta\geq1$, we define $T+1$ degree classes, where $T=O\big(\frac{\log(\Delta+1)}{\eps}\big)$ (and $T=1$ when $\Delta=1$), as follows.

\begin{itemize}
    \item Set $d_0:=0$ and, for $i\geq1$, set $d_i:=\min\{\Delta,\lfloor(d_{i-1}+1)/(1-\eps')\rfloor\}$, stopping at the first index $T$ with $d_T=\Delta$.
    \item Partition the nodes in $V$ into sets $V_0,V_1,\dots,V_T$ such that $V_0$ consists of the nodes of degree $d_0=0$ and for $i\geq 1$, $V_i:=\set{v \in V\,:\, \deg(v)\in [d_{i-1}+1,d_i]}$.
\end{itemize}
Before the cap at $\Delta$ is reached, $d_i+1>(d_{i-1}+1)/(1-\eps')$ and $d_i\geq d_{i-1}+1$. Hence the intervals have no gaps or repetitions, and $T\leq1+\log(\Delta+1)/[-\log(1-\eps')]=O(\log(\Delta+1)/\eps)$.

Now, we can start running the following procedure. We iteratively compute a sequence of independent sets $I_0, I_1, \dots, I_T$ and node weight functions $w_0, w_1, \dots, w_T$. Initially for all $v \in V$, let $w_0(v)=w(v)$ and $I_0=V_0$. Note that $I_0$ is an independent set of $G$ because the nodes in $V_0$ have degree $d_0=0$. For $i\geq 1$, we then do the following.

\begin{enumerate}
    \item Define node weights $w_i:V\to \mathbb{Z}_{\geq0}$ such that for all $v\in V$, $w_i(v) := \max\big\{0, w_{i-1}(v) - \sum_{u\in N^+(v)\cap I_{i-1}} w_{i-1}(u)\big\}$.
    \item Use \Cref{lemma:Deltaapprox} on $G[V_0\cup\dots\cup V_i]$, whose maximum degree is at most $d_i$, and discard any returned zero-weight nodes, obtaining an independent set $I_i$ whose nodes have positive $w_i$-weight and such that $w_i(I_i) \geq (1-\eps')\cdot w_i(V_0\cup\dots\cup V_i)/(d_i+1)$.
\end{enumerate}

\noindent The sets $I_0,\dots,I_T$ and weight functions $w_0,\dots,w_T$ satisfy the conditions of \Cref{lemma:localratioseq} and given $I_0,\dots,I_T$ we can therefore use \Cref{lemma:localratioseq} to compute an independent set $I$ of weight $w(I)\geq\sum_{i=0}^T w_i(I_i)$.

We next prove \Cref{thm:carowei} from the algorithm above. The following is a slightly more general restatement of \Cref{thm:carowei}. Note that \Cref{thm:carowei} directly follows from \Cref{thm:carowei2} by setting $q=\xi$ and using the supplied proper $\xi$-coloring as the initial $q$-coloring. The proof of \Cref{thm:carowei2} appears below.

\begin{theorem}\label{thm:carowei2}
    There exists a deterministic
    $O\big(\frac{\log(1/\eps)}{\eps}\cdot \log^3\Delta + \log^* q\big)$-round \CONGEST algorithm that, for every graph $G=(V,E)$ with node weights $w:V\to\mathbb{Z}_{\geq0}$, maximum degree $\Delta$, an initial proper $q$-coloring with $q\leq n^{O(1)}$, and for any $0<\eps\leq 1/2$,
    computes an independent set $I$ of total weight $w(I) \geq (1-\eps)\cdot\sum_{v\in V}\frac{w(v)}{\deg(v)+1}$.
\end{theorem}

\begin{proof}
    If $\Delta=0$, return $I=V$ locally. Hence, assume $\Delta\geq1$. The independent set $I$ is computed in the following way. As a first step, we use the algorithm of \cite{linial92} to compute a proper $O(\Delta^2)$-coloring of $G$ in $O(\log^* q)$ rounds. We then compute the independent sets $I_0,\dots,I_T$ and the weight functions $w_0,\dots,w_T$ as described above. We then use \Cref{lemma:localratioseq} to compute an independent set $I$ of total weight
    \begin{equation}\label{eq:caroweiTotal}
        w(I) \geq \sum_{i=0}^T w_i(I_i).
    \end{equation}
    We first analyze the time complexity. Computing $I_0,\dots,I_T$ requires running the algorithm of \Cref{lemma:Deltaapprox} $T=O(\log(\Delta+1)/\eps)$ times. Given $I_{i-1}$, the weight function $w_i$ can be computed in one \CONGEST round. The overall time to compute these sets and weights is therefore
    \[
    T\cdot O\big(\log^2\Delta\cdot\log(1/\eps) + \log^* \Delta\big) =
    O\left(\frac{\log(1/\eps)}{\eps}\cdot \log^3\Delta\right).
    \]
    Since by \Cref{lemma:localratioseq}, the final independent set $I$ can then be computed in another $O(T)$ rounds, together with the $O(\log^* q)$ rounds required to compute the proper $O(\Delta^2)$-coloring, the claimed time complexity follows.
    
    We next prove the weight guarantee. For every $i\in\{0,\dots,T\}$, let $U_i:=V_0\cup\dots\cup V_i$ and define
    \[
        \Phi_i:=\sum_{j=0}^{i}(d_j+1)w_j(I_j).
    \]
    We first show that
    \begin{equation}\label{ineq:phi}
        \Phi_i\geq(1-\eps')w(U_i).
    \end{equation}
    The claim holds for $i=0$ because $I_0=V_0$ and hence $\Phi_0=w(V_0)$. Now consider $i\geq1$. When a node $v\in I_j$ is processed, it decreases the residual weight of at most the nodes in $N^+(v)$. Since $I_j\subseteq U_j$, we have $\deg(v)\leq d_j$, and therefore $v$ causes a total decrease of at most
    \[
        (\deg(v)+1)w_j(v)\leq(d_j+1)w_j(v).
    \]
    Summing this bound over the previously processed independent sets gives the charging inequality
    \begin{equation}\label{ineq:residual-prefix}
        w_i(U_i)\geq w(U_i)-\Phi_{i-1}.
    \end{equation}
    The guarantee for $I_i$ and \eqref{ineq:residual-prefix} now imply
    \[
    \begin{aligned}
        \Phi_i
        &=\Phi_{i-1}+(d_i+1)w_i(I_i)\\
        &\geq\Phi_{i-1}+(1-\eps')w_i(U_i)\\
        &\geq(1-\eps')w(U_i)+\eps'\Phi_{i-1}\\
        &\geq(1-\eps')w(U_i),
    \end{aligned}
    \]
    which proves \eqref{ineq:phi}.

    It remains to translate the lower bounds on the prefix potentials into a lower bound on the returned independent set. Write $a_i:=w_i(I_i)$. By the definition of $\Phi_i$, we have $a_0=\Phi_0/(d_0+1)$ and $a_i=(\Phi_i-\Phi_{i-1})/(d_i+1)$ for $i\geq1$. Consequently,
    \[
    \begin{aligned}
        w(I)
        &\stackrel{\eqref{eq:caroweiTotal}}{\geq}\sum_{i=0}^T a_i\\
        &=\frac{\Phi_T}{d_T+1}
          +\sum_{i=0}^{T-1}
            \left(\frac1{d_i+1}-\frac1{d_{i+1}+1}\right)\Phi_i.
    \end{aligned}
    \]
    This is a summation-by-parts identity. All its coefficients are nonnegative because the thresholds increase. Applying \eqref{ineq:phi}, we obtain
    \[
        w(I)\geq(1-\eps')X,
    \]
    where
    \[
    \begin{aligned}
        X
        &:=\frac{w(U_T)}{d_T+1}
          +\sum_{i=0}^{T-1}
            \left(\frac1{d_i+1}-\frac1{d_{i+1}+1}\right)w(U_i)\\
        &=\frac{w(V_0)}{d_0+1}
          +\sum_{i=1}^T\frac{w(U_i)-w(U_{i-1})}{d_i+1}\\
        &=\sum_{i=0}^T\frac{w(V_i)}{d_i+1}.
    \end{aligned}
    \]
    We therefore have
    \[
    w(I) \geq (1-\eps')\cdot \sum_{i=0}^T \frac{w(V_i)}{d_i+1} \geq
    (1-\eps')\cdot \left(w(V_0) + \sum_{i=1}^T\sum_{v\in V_i}
    \frac{(1-\eps')\cdot w(v)}{d_{i-1}+2}\right).
    \]
    The second inequality uses $d_i/(d_{i-1}+1)\leq1/(1-\eps')$, and hence $d_i+1\leq(d_{i-1}+2)/(1-\eps')$. Since $\deg(v)=0$ for $v\in V_0$ and $\deg(v)\geq d_{i-1}+1$ for $v\in V_i$ and $i\geq1$, we obtain
    \[
    w(I) \geq
    (1-\eps')^2\cdot\sum_{v\in V}\frac{w(v)}{\deg(v)+1} \geq
    (1-\eps)\cdot \sum_{v\in V}\frac{w(v)}{\deg(v)+1}.
    \]
    The last inequality uses that $(1-\eps')^2\geq 1-2\eps'=1-\eps$.
\end{proof}

\subsection{Better Approximation as a Function of the Arboricity}

We now give an algorithm to compute an independent set of weight at least $(1-\eps)w(V)/(2\beta+1)$ (\Cref{thm:arboricity}). We use an internal accuracy parameter $0<\eta\leq1/2$, setting $\eta=\eps/2$ for the theorem, and define $A:=\lfloor(2+\eta)\beta\rfloor$. Throughout the recursion, $\beta$, $A$, and $\eta$ remain fixed. A proper coloring $\varphi$ is computed once as described below and restricted to each recursive graph. All degrees and neighborhoods in a call refer to its current graph $G$. Because the average degree is less than $2\beta$, an $\Omega(\eta)$-fraction of the nodes have degree at most $A$. We compute an independent set $I_0$ among these nodes, adjust the remaining weights using \Cref{lemma:localratio}, recurse, and combine the two independent sets. \Cref{algo:2beta+1approxMWIS} gives the details.

\begin{algorithm}[H]\caption{Recursive algorithm for computing an independent set in a graph $G$ of arboricity $\beta$ with node weights $w$}\label{algo:2beta+1approxMWIS}
    \begin{algorithmic}[1]
		\Function{ApproxMWIS}{$G=(V,E)$, $w$, $\beta$, $\eta$, $\varphi$}
        \State \textbf{if} $V=\emptyset$ \textbf{then return} $\emptyset$
        \State $V_0 := \set{v\in V\,:\,\deg(v)\leq A}$
        \For{\textbf{all} $v \in V_0$} 
        \State Define $\overline{w}(v):= w(v)\cdot(\deg_0(v)+1)$, where $\deg_0(v)$ is the degree of $v$ in $G[V_0]$
        \EndFor
        \State \parbox[t]{\dimexpr\linewidth-\algorithmicindent\relax}{Run the algorithm of Thm.\ \ref{thm:carowei2} on $G[V_0]$ with weights $\overline{w}$, accuracy $\eta$, and coloring $\varphi|_{V_0}$ to get $I_0$}\label{line:arboricity-carowei}
        \For{\textbf{all} $v \in V\setminus V_0$}
        \State \reviewnote[-8pt]{Use closed neighborhood.} Define $w'(v) := \max\set{0, w(v) - \sum_{u\in N^{\reviewchange{+}}(v)\cap I_{0}} w(u)}$
        \EndFor
        \State Independent set $I':=$ \textsc{ApproxMWIS}$\big(G[V\backslash V_0],w',\beta,\eta,\varphi|_{V\setminus V_0}\big)$\label{line:arboricity-recursion}
        \State Remove nodes $v$ with $w'(v)=0$ from $I'$
        \State \reviewnote[5pt]{Disjoint sets use $N$.} \Return $I:=I' \cup (I_0 \backslash \reviewchange{N(I')})$
        \EndFunction
	\end{algorithmic}		
\end{algorithm}

We first analyze the round complexity.

\begin{lemma}\label{lemma:2beta_timecomplexity}
  \Cref{algo:2beta+1approxMWIS} can be implemented on any node-weighted graph $G=(V,E)$ of arboricity at most $\beta\geq1$ and any $0<\eta\leq1/2$ in $O\big(\frac{\log^3(\beta)\cdot\log(1/\eta)\cdot\log n}{\eta^2}\big)$ rounds in the (deterministic) \CONGEST model.
\end{lemma}
\begin{proof}
    First, we compute a proper $O(\beta^2)$-coloring of $G$ deterministically in $O(\log n)$ rounds using Theorem 5.5 of \cite{barenboim08}. The construction exchanges identifiers, colors, and decomposition information with $O(\log n)$-bit encodings, and therefore also runs in \CONGEST. This coloring is computed once and restricted in recursive calls.

    We analyze the time per recursion level and the recursion depth. At a fixed level, determining $V_0$, exchanging membership information, and computing the residual weights take $O(1)$ rounds outside the call to \Cref{thm:carowei2}. The time for Line~\ref{line:arboricity-carowei} is
    $O\big(\frac{\log(1/\eta)}{\eta}\cdot \log^3 \beta\big)$.
    Here, $G[V_0]$ has maximum degree at most $A=O(\beta)$ and a proper $O(\beta^2)$-coloring, so the $\log^* q$ term in \Cref{thm:carowei2} is absorbed.

    To bound the recursion depth, note that $G$ has average degree at most $2\beta$ and every node in $V\setminus V_0$ has degree at least $(2+\eta)\beta$. Hence, $|V\setminus V_0|\cdot (2+\eta)\beta \leq 2\beta\cdot n$, and thus $|V\setminus V_0|\leq \frac{2}{2+\eta}\cdot n = (1-\Theta(\eta))\cdot n$. Thus, the number of nodes decreases by a factor $(1-\Theta(\eta))$ per level, giving a recursion depth of $O\big(\frac{\log n}{\eta}\big)$.

    We implement the recursion with the fixed schedule
    $L:=1+\lceil\log n/\log(1+\eta/2)\rceil=O(\log n/\eta)$
    of forward levels, each padded to the common per-level runtime bound, followed by $L$ reverse composition rounds. Inactive nodes wait. Thus, no node needs to detect globally whether the remaining graph is empty. Combining a level's two independent sets takes one reverse round. Multiplying the per-level bound by $L$ yields the claimed time complexity.
\end{proof}

It remains to lower-bound the weight of the independent set computed by \Cref{algo:2beta+1approxMWIS}. The proof of the following lemma, by induction on $|V|$, appears below.

\begin{lemma}\label{lemma:approxRatio_arboricity}
\Cref{algo:2beta+1approxMWIS}, when applied to any node-weighted graph $G=(V,E)$ of arboricity at most $\beta$ with node weights $w:V\to\mathbb{Z}_{\geq0}$, for any parameter $0<\eta\leq 1/2$, and for $A:=\lfloor(2+\eta)\beta\rfloor$, returns an independent set $I$ of weight
    \[
        w(I) \geq (1-\eta)\cdot\frac{w(V)}{A+1} \geq \frac{1-\eta}{1+\eta/2}\cdot\frac{w(V)}{2\beta+1} \geq (1-2\eta)\cdot \frac{w(V)}{2\beta+1}.
    \]
\end{lemma}

\begin{proof}
    Let $A:=\lfloor(2+\eta)\beta\rfloor$ and $c:=(1-\eta)/(A+1)$. We prove $w(I)\geq c\,w(V)$ by induction on $|V|$, keeping $\beta$, $A$, and $\eta$ fixed throughout the recursion. If $V=\emptyset$, the algorithm returns $\emptyset$ and the bound holds. Otherwise, $G$ has a node of degree at most $A$, so $V_0\neq\emptyset$. Hence, $R:=V\setminus V_0$ has fewer nodes than $G$, and $G[R]$ still has arboricity at most $\beta$.

    Write $\deg_0(v):=\deg_{G[V_0]}(v)$. Applying \Cref{thm:carowei2} to $G[V_0]$ with weights $\overline{w}(v)=(\deg_0(v)+1)w(v)$ in Line~\ref{line:arboricity-carowei} gives
    \begin{equation}\label{eq:arboricity-carowei}
        \sum_{u\in I_0}(\deg_0(u)+1)w(u)
        = \overline{w}(I_0)
        \geq (1-\eta)\sum_{v\in V_0}\frac{\overline{w}(v)}{\deg_0(v)+1}
        = (1-\eta)w(V_0).
    \end{equation}
    Each $u\in I_0$ has at most $A-\deg_0(u)$ neighbors in $R$. The definition of the residual weights therefore implies
    \begin{equation}\label{eq:arboricity-residual}
        w'(R)\geq w(R)-\sum_{u\in I_0}(A-\deg_0(u))w(u).
    \end{equation}
    By the induction hypothesis, the recursive call returns an independent set of residual weight at least $c\,w'(R)$. Removing nodes of residual weight zero preserves this bound and ensures that \Cref{lemma:localratio} applies to the returned set $I'$. Consequently,
    \begin{equation}\label{eq:boundonwI}
        \begin{aligned}
            w(I)&\geq w(I_0)+c\,w'(R)\\
                &\geq c\,w(R)+\sum_{u\in I_0}\bigl[1-c(A-\deg_0(u))\bigr]w(u)\\
                &\geq c\,w(R)+\frac{1}{A+1}\sum_{u\in I_0}(\deg_0(u)+1)w(u)\\
                &\geq c\,w(V).
        \end{aligned}
    \end{equation}
    The second inequality uses \eqref{eq:arboricity-residual} and $c\geq0$. The third follows from
    \[
        1-c(A-\deg_0(u))-\frac{\deg_0(u)+1}{A+1}
        =\frac{\eta(A-\deg_0(u))}{A+1}\geq0,
    \]
    and the last follows from \eqref{eq:arboricity-carowei}. This completes the induction.

    Finally, since
    \[
        A+1\leq (1+\eta/2)(2\beta+1)
        \qquad\text{and}\qquad
        \frac{1-\eta}{1+\eta/2}\geq 1-2\eta,
    \]
    the remaining inequalities in the lemma follow, including when $w(V)=0$.
\end{proof}

\Cref{thm:arboricity} follows from \Cref{lemma:2beta_timecomplexity,lemma:approxRatio_arboricity} with $\eta=\eps/2$.


\clearpage
\bibliographystyle{alpha}
\bibliography{references}

\clearpage

\appendix
\section{Proof of \Cref{thm:upperOther}}\label{app:upperOther}
\begin{proof}
    Let $T=(V,E)$ be the input tree. Trees with one or two nodes can be solved exactly in $O(1)$ rounds, using IDs to break ties. Hence, assume $|V|\geq3$. Let $V'$ be the set of non-leaf nodes, let $m:=|V'|\geq1$, and let $T':=T[V']$. The graph $T'$ is a tree: every internal node on the path between two non-leaves is itself a non-leaf. Every node in $V\setminus V'$ is a leaf whose unique neighbor belongs to $V'$. For a graph $H$, write $\gamma(H)$, $\tau(H)$, and $\nu(H)$ for its minimum dominating set size, minimum vertex cover size, and maximum matching size, respectively.

    \emph{Lower bounds on the optima.}
    For the analysis only, root $T$ at any node of $V'$. Each node in $V'$ has a child; choose one such child for each of these nodes. The $m$ chosen parent-to-child edges form a forest of paths, since each node has at most one chosen edge to a child and at most one chosen edge to its parent. Taking alternate edges on each path gives a matching of size at least $m/2$. Every vertex cover must contain an endpoint of each edge of this matching, so
    \[
        \nu(T)\geq m/2,\qquad \tau(T)\geq m/2.
    \]
    Let $S$ be the set of chosen children. Each node has a unique parent, so $|S|=m$. For every node $v$, the closed neighborhood $N^+(v)$ contains at most three nodes of $S$: possibly $v$ itself, its parent, and its one chosen child. Every dominating set $D$ covers $S$, and therefore $m\leq\sum_{v\in D}|N^+(v)\cap S|\leq3|D|$. Thus $\gamma(T)\geq m/3$, and all three optima are at least $m/3$. This rooting and the chosen edges are only used to prove these bounds and need not be computed by the algorithm.

    \emph{Clustering.}
    Give each edge of $T'$ weight one and apply \Cref{thm:clustering} with parameter $\delta=\eps/3$, using the original IDs and network size bound. If $m=1$, use the single-node cluster directly. Attach each leaf of $T$ to the cluster of its neighbor. This increases cluster diameters by at most two and creates no additional intercluster edges. Consequently, each cluster has diameter $O(1/\eps)$, and the set $E_{\mathrm{cut}}$ of intercluster edges satisfies
    \[
        b:=|E_{\mathrm{cut}}|\leq\frac{\eps}{3}|E(T')|
          =\frac{\eps}{3}(m-1)\leq\frac{\eps m}{3}.
    \]
    Let $F:=(V,E\setminus E_{\mathrm{cut}})$. Its components are precisely the cluster trees. Identifying non-leaves and attaching leaves takes $O(1)$ rounds, so the clustering phase takes $O(\log^*(n)/\eps)$ rounds in \CONGEST.

    \emph{Exact solutions within clusters.}
    We next describe how to solve all three problems exactly on each cluster tree in parallel. A tree of diameter $D$ can be rooted in $O(D+1)$ rounds by synchronous leaf peeling until one or two centers remain, breaking a tie by ID, and broadcasting the chosen root. Root each cluster this way. Its height is $h=O(1/\eps)$. Let $C(v)$ denote the children of $v$, and let $T_v$ be the subtree rooted at $v$. All the following dynamic programs send a constant number of values per edge from children to parents. Empty sums are zero, and a minimum over an empty set is $+\infty$.

    For MDS, let $A_v$, $B_v$, and $C_v$ be the minimum sizes of sets that dominate all nodes of $T_v\setminus\{v\}$ and satisfy, respectively: $v$ is selected; $v$ is not selected but is dominated by a selected child; or $v$ is not selected and is not dominated within $T_v$, so it must be dominated by its parent. The recurrences are
    \[
        \begin{aligned}
            A_v&=1+\sum_{u\in C(v)}\min\{A_u,B_u,C_u\},\\
            B_v&=\min_{z\in C(v)}\left(A_z+
                \sum_{u\in C(v)\setminus\{z\}}\min\{A_u,B_u\}\right),\\
            C_v&=\sum_{u\in C(v)}B_u.
        \end{aligned}
    \]
    If $v$ is selected, a child may rely on $v$ for domination. If $v$ is not selected, no child may rely on it; moreover, state $B_v$ requires at least one selected child, whereas state $C_v$ permits none. These observations prove the recurrences by induction. A leaf has values $(A_v,B_v,C_v)=(1,+\infty,0)$, and the optimum at a cluster root $r$ is $\min\{A_r,B_r\}$, since the root cannot rely on a parent.

    For MVC, let $K_v^1$ and $K_v^0$ be the minimum vertex cover sizes in $T_v$ subject to selecting or not selecting $v$. Then
    \[
        K_v^1=1+\sum_{u\in C(v)}\min\{K_u^0,K_u^1\},
        \qquad K_v^0=\sum_{u\in C(v)}K_u^1.
    \]
    Indeed, if $v$ is not selected, every child must be selected to cover its edge to $v$; if $v$ is selected, either state of each child is allowed. The leaf values are $(K_v^1,K_v^0)=(1,0)$, and the optimum at the root is $\min\{K_r^0,K_r^1\}$.

    For MM, let $M_v^1$ and $M_v^0$ be the maximum numbers of matching edges inside $T_v$ subject to $v$ being matched or not matched to its parent. The parent edge itself is not counted. Thus
    \[
        \begin{aligned}
            M_v^1&=\sum_{u\in C(v)}M_u^0,\\
            M_v^0&=\max\left\{\sum_{u\in C(v)}M_u^0,
                \max_{z\in C(v)}\left(1+M_z^1+
                    \sum_{u\in C(v)\setminus\{z\}}M_u^0\right)\right\}.
        \end{aligned}
    \]
    In state $M_v^1$, $v$ cannot be matched to any child. In state $M_v^0$, it is either unmatched or matched to exactly one child $z$. Both leaf values are zero, the second alternative is absent when $v$ has no children, and the optimum at the root is $M_r^0$.

    Each node stores choices attaining its state values, breaking ties deterministically by IDs. After the bottom-up computation, the root chooses an optimal permitted state and a top-down pass tells each child its required state. This reconstructs an optimal solution, including both endpoints' knowledge of each matching edge. Every finite value is an integer between zero and $n$; infeasibility is represented by a separate flag. Hence each message has $O(\log n)$ bits, and the two passes take $O(h+1)=O(1/\eps)$ rounds. Distinct clusters use disjoint edges, so all these computations run concurrently in \CONGEST.

    \emph{Approximation guarantees.}
    The cluster solutions give optimal solutions on $F$. Let $D_F$, $K_F$, and $M_F$ denote the resulting dominating set, vertex cover, and matching.

    The set $D_F$ also dominates $T$. Conversely, let $D$ be an optimal dominating set of $T$. For every cut edge with exactly one endpoint in $D$, add the other endpoint to $D$. This adds at most $b$ nodes and gives a dominating set of $F$: any node that loses all its neighbors in $D$ when cut edges are removed is itself added. Consequently,
    \[
        |D_F|=\gamma(F)\leq\gamma(T)+b\leq(1+\eps)\gamma(T).
    \]
    Every vertex cover of $T$ covers $F$, so $|K_F|=\tau(F)\leq\tau(T)$. For each cut edge not covered by $K_F$, add its endpoint with smaller ID. The resulting set $K$ covers $T$ and satisfies
    \[
        |K|\leq\tau(F)+b\leq\tau(T)+b\leq(1+\eps)\tau(T).
    \]
    Finally, $M_F$ is a matching in $T$. Removing the cut edges from an optimal matching of $T$ leaves a matching in $F$ with at least $\nu(T)-b$ edges. Therefore,
    \[
        |M_F|=\nu(F)\geq\nu(T)-b\geq(1-\eps)\nu(T).
    \]
    The final inequalities in all three cases use $b\leq\eps m/3$ and the lower bounds on the optima. The vertex-cover repair takes $O(1)$ rounds, including exchanging cluster IDs and membership bits across cut edges. Together with clustering and the cluster dynamic programs, the total deterministic round complexity is $O(\log^*(n)/\eps)$.
\end{proof}

\end{document}